\documentclass[a4paper,UKenglish,numberwithinsect,nolineno]{lipics-v2019}

\usepackage{microtype}
\usepackage{tikz}
\usetikzlibrary{backgrounds}
\usetikzlibrary{calc}
\usetikzlibrary{decorations.pathmorphing}

\tikzstyle{labeled}=[fill=white,draw=black,circle,inner sep=0, minimum size=5.6mm, thin]
\tikzstyle{simple}=[circle, draw=black, fill=white, inner sep=0, minimum size=2.5mm, thin]
\tikzstyle{edge-label}=[circle, fill=white, inner sep=.16mm]

\usepackage{todonotes}
\newcommand{\doneTodo}[1]{}
\newcommand{\laterTodo}[1]{}

\newtheorem{problem}[theorem]{Problem}
\newtheorem{observation}[theorem]{Observation}

\usepackage{tikz}
\usetikzlibrary{shapes}
\tikzstyle{every picture} = [>=latex]

\def\ca#1{{\cal#1}}

\newcommand{\dfnc}[3]{#1:#2\rightarrow #3} 
\newcommand{\DEFI}[1]{\emph{#1}}
\newcommand{\crn}{\operatorname{cr}}
\newcommand{\lset}[2]{\{#1,\ldots,#2\}}
\newcommand{\dset}[2]{\{#1\;|\;#2\}}
 \nolinenumbers

\newcommand{\ccgn}[2]{\ensuremath{G_{#1}^{(#2)}}}

\definecolor{darkgreen}{rgb}{0,0.5,0}

\title{Bounded degree conjecture holds precisely for $c$-crossing-critical graphs with $c\le 12$}

\titlerunning{Bounded maximum degree in crossing-critical graphs}

\author{Drago Bokal}{Department of Mathematics and Computer Science, University of Maribor, Maribor, Slovenia}{drago.bokal@um.si}{0000-0003-1196-5476}{ARRS Project J1-8130, ARRS Programme 	P1-0297}

\author{Zden\v{e}k Dvo\v{r}\'ak}{Computer Science Institute, Charles University, Prague, Czech Republic}{rakdver@iuuk.mff.cuni.cz}{0000-0002-8308-9746}{Supported by project 17-04611S (Ramsey-like aspects of graph coloring) of Czech Science Foundation.}

\author{Petr Hlin\v{e}n\'{y}}{Faculty of Informatics, Masaryk University, Brno, Czech Republic}
	{hlineny@fi.muni.cz}{0000-0003-2125-1514}
	{Supported by the Czech Science Foundation, projects no.~17-00837S and~20-04567S.}

\author{Jes\'{u}s Lea\~{n}os}{Academic Unit of Mathematics, Autonomous University of Zacatecas, Zacatecas,  Mexico}{jleanos@matematicas.reduaz.mx}{0000-0002-9068-8536}{Partially supported by PFCE-UAZ 2018-2019 grant.}

\author{Bojan Mohar}
    {Department of Mathematics, Simon Fraser University, 
        Burnaby BC, Canada and\\ Institute of Mathematics, Physics, and Mechanics, Ljubljana, Slovenia}{mohar@sfu.ca}{0000-0002-7408-6148}{B.M.~was supported in part by the NSERC Discovery Grant R611450
 	(Canada), by the Canada Research Chairs program,
	and by the Research Project J1-8130 of ARRS (Slovenia).}

\author{Tilo Wiedera}{Theoretical Computer Science, Osnabrück University, Germany}{tilo.wiedera@uos.de}{0000-0002-5923-4114}{Supported by the German Research Foundation (DFG) project CH 897/2-2.}

\authorrunning{Bokal, Dvo\v{r}\'ak, Hlin\v{e}n\'{y}, Lea\~{n}os, Mohar, Wiedera}

\Copyright{Drago Bokal, Zden\v{e}k Dvo\v{r}\'{a}k, Petr Hlin\v{e}n\'{y}, Jes\'{u}s Lea\~{n}os, 
Bojan Mohar, Tilo Wiedra}



\category{}

\relatedversion{This work is based on a SoCG contribution \cite{bokal2019bounded} and gives full proofs and a significant strengthening of the announced results.}

\supplement{}

\funding{}

\acknowledgements{We acknowledge the supporting environment of the workshop 
	Crossing Numbers: Theory and Applications (18w5029)
	at the Banff International Research Station, where the fundamentals of this contribution were developed.}

\hideLIPIcs

\begin{document}
\maketitle

\begin{abstract}
We study $c$-crossing-critical graphs, which are the minimal graphs that
require at least $c$ edge-crossings when drawn in the plane.
For every fixed pair of integers with $c\ge 13$ and $d\ge 1$, 
we give first explicit constructions of $c$-crossing-critical graphs 
containing arbitrarily many vertices of degree greater than $d$. 
We also show that such unbounded degree constructions do not exist for $c\le 12$,
precisely, that there exists a constant $D$ such that every
$c$-crossing-critical graph with $c\le 12$ has maximum degree at most~$D$.
Hence, the bounded maximum degree conjecture 
of $c$-crossing-critical graphs, which was generally disproved in 2010
by Dvo\v{r}\'ak and Mohar (without an explicit construction), holds true,
surprisingly, exactly for the values $c\le 12.$
\end{abstract}

\section{Introduction}
\doneTodo{Referee 3: 
LL 12,14:  The meaning of "cf." is "confer" (compare), usually used to
  highlight two *opposite* ideas.  It should not be confused with "e.g."
  (see for example). \\
Response: Replaced ``cf." with ``see´´. (Drago)}
\doneTodo{Referee 4: 
line 14: Please, explain the term "sparse graph".\\
Response: PH done more carefully (not to confuse with "sparsity").}
Minimizing the number of edge-crossings in a graph drawing in the plane
(the \DEFI{crossing number} of the graph, see~Definition~\ref{def:crossingn}) 
is considered one of the most important attributes of a ``nice drawing'' of a graph.
In the case of classes of dense graphs (those having superlinear number of edges in
terms of the number vertices), 
the crossing number is necessarily very high -- see the famous 
Crossing Lemma \cite{ajtaiChvatalNewbornSzemeredi82,leighton83}.
However, within sparse graph classes (those having only linear number of edges), 
we may have planar graphs at one end and graphs with up
to quadratic crossing number at the other end.
In this situation, it is natural to study the ``minimal obstructions'' for
low crossing number, with the following definition.

Let $c$ be a positive integer. A graph $G$ is called \DEFI{$c$-crossing-critical} 
if the crossing number of $G$ is at least $c$, but every proper subgraph 
has crossing number smaller than $c$. We say that $G$ is \DEFI{crossing-critical}
if it is $c$-crossing-critical for some positive integer $c$. 

Since any non-planar graph contains at least one crossing-critical subgraph,
the understanding of the properties of the crossing-critical graphs is a central
part of the theory of crossing numbers. 

In 1984, \v{S}ir\'a\v{n} gave the earliest construction of 
nonsimple $c$-critical-graphs for every fixed value of $c\ge 2$ \cite{siran84}.
Three years later, Kochol \cite{kochol87} gave an infinite family of
c-crossing-critical, simple, 3-connected graphs, for every $c\ge 2$. Another early 
result on $c$-crossing-critical graphs was reported in the influential paper of 
Richter and Thomassen \cite{richterThomassen93}, who proved that $c$-crossing-critical
graphs have bounded crossing number in terms of $c$.
\doneTodo{Referee 4: 
line 28: ... have bounded crossing number. In terms of what? \\
Response: in terms of $c$. Done. (Drago)}
They also initiated research on degrees in $c$-crossing-critical 
graphs by showing that, if there exists an infinite family of 
$r$-regular, $c$-crossing-critical graphs for fixed $c$, 
then $r\in\{4,5\}$. Of these, $4$-regular $3$-critical graphs 
were constructed by Pinontoan and Richter \cite{pinontoanRichter03},
and $4$-regular $c$-critical graphs are known for every $c\ge 3$, 
$c\ne 4$ \cite{bokalBracicDernarHlineny19}.
Salazar observed that the arguments of Richter and Thomassen
could be applied to average degree as well, showing that an
infinite family of $c$-crossing-critical graphs of average degree $d$ 
can exist only for $d\in (3,6]$, and established their existence
for $d\in [4,6)$. Nonexistence of such families with $d=6$ was established much later by Hern\'andez, Salazar,
and Thomas \cite{velezSalazarThomas12}, who proved that, for each fixed $c$,
there are only finitely many $c$-crossing-critical simple graphs of average degree 
at least six. The existence of such families with $d\in [\tfrac{7}{2},4]$ 
was established by Pinontoan and Richter \cite{pinontoanRichter03}, whereas
the whole possible interval was covered by Bokal \cite{bokal10}, who showed
that, for sufficiently large crossing number, both the crossing number $c$ 
and the average degree $d\in (3,6)$ could be prescribed for an infinite family of 
$c$-crossing critical graphs of average degree $d$. 

In 2003, Richter conjectured that, for every positive integer $c$, there 
exists an integer $D(c)$ such that every $c$-crossing-critical graph  
has maximum degree fewer than $D(c)$~\cite{moharNowakowskiWest07}.
Reflecting upon this conjecture, Bokal in 2007 observed that the known 
$3$-connected crossing-critical graphs 
of that time only had degrees $3,4,6$, and asked for existence of such graphs
with arbitrary other degrees, possibly appearing arbitrarily many times.
Hlin\v{e}n\'y augmented his construction of $c$-crossing-critical graphs with
pathwidth linear in $c$ to show the existence of $c$-crossing-critical graphs with
\doneTodo{Referee 4: 
line 47: What is "high in c"?\\
Response: PH done.}
arbitrarily many vertices of every set of even degrees. Only a recent paper 
by Bokal, Bračič, Der\v{n}\'ar, and Hlin\v{e}n\'{y}~\cite{bokalBracicDernarHlineny19} provided the corresponding
result for odd degrees, showing in addition that, for sufficiently high $c$,
all the three parameters -- crossing number~$c$, rational average degree $d$, and
the set of degrees $D\subseteq \mathbb N\setminus \{1,2\}$
that appear arbitrarily often in the graphs of the infinite family -- 
can be prescribed. They also analysed the interplay of these
parameters for $2$-crossing-critical graphs that were recently 
completely characterized by Bokal, Oporowski, Richter, and Salazar \cite{bokalOporowskiRichter16}.

Despite all this research generating considerable understanding of the 
behavior of degrees in known crossing-critical graphs as well as extending
the construction methods of such graphs, the original conjecture of Richter 
was not directly addressed in the previous works. It was, however, disproved by 
Dvo\v{r}\'ak and Mohar \cite{dvorakMohar10}, who showed that, for each integer 
$c\geq 171$, there exist $c$-crossing-critical graphs of arbitrarily large 
maximum degree. 
Their counterexamples, however,
were not constructive, as they only exhibited, for every such $c$, a graph 
containing sufficiently many  critical edges incident with a fixed vertex and 
argued that those edges belong to every $c$-crossing-critical subgraph of the 
exhibited graph. On the other hand, as a consequence of \cite{bokalOporowskiRichter16}
it follows that, except for possibly some small examples, the maximum degree in a 
large $2$-crossing-critical graph is at most 6, implying that Richter's conjecture holds
for $c=2$. In view of these results, and the fact that $1$-crossing-critical graphs
(subdivisions of $K_5$ and $K_{3,3}$) have maximum degree at most $4$, this leaves
Richter's conjecture unresolved for each $c\in\{3,4,\ldots,170\}.$ 

The richness of $c$-crossing-critical graphs is restricted for every $c$
by the result of Hlin\v{e}n\'{y} that $c$-crossing-critical 
graphs have bounded path-width \cite{hlineny03}; 
this structural result is complemented by a recent classification of
all large $c$-crossing-critical graphs for arbitrary $c$ by 
Dvo\v{r}\'ak, Hlin\v{e}n\'{y}, and Mohar \cite{dvorakHlinenyMohar18}.
We use these
results in Section \ref{sc:boundedMD} to show that Richter's conjecture holds for
$c\le 12$. The result is stated below. It is both precise and surprising and shows how unpredictable 
are even the most fundamental questions about crossing numbers.
\doneTodo{Referee 2,3: 
L72:  “results ARE stated”.\\
Response: Singular noun is used rather than plural verb, as there is only one result in question. (Drago)}

\begin{theorem}
\label{th:boundedMD}
There exists an integer $D$ such that, for every positive integer $c\le 12$,
every $c$-crossing-critical graph has maximum degree at most $D$.
\end{theorem}
\doneTodo{Referee 1: Theorem 1.1: What bounds do we get for D?
(either as a ``uniform´´ D, or even better, as a function of c)
I think, it should definitely be mentioned, together with the best known
lower bounds.\\
Response: PH - we have nothing, right? Even $c=2$? I will start to investigate some
precise values with my student.\\
Drago: right, excellent for your student! We can discuss this and some other ideas in Telč.}
In fact, one can separately consider in Theorem~\ref{th:boundedMD} twelve
upper bounds $D_c$ for each of the values $c\in\{1,2,\ldots,12\}$.
For instance, $D_1=4$ and the optimal value of $D_2$ (we know $D_2\geq8$)
should also be within reach using \cite{bokalOporowskiRichter16} and
continuing research.
On the other hand, due to the asymptotic nature of our arguments,
we are currently not able to give any ``nice'' numbers for the remaining
upper bounds, and we leave this aspect to future investigations.

We cover the remaining values of $c\ge 13$ in the gap
in a very strong sense, by constructing critical graphs with arbitrarily
many high-degree vertices:
\begin{theorem}
\label{th:construction13}
For every positive integers $d$ and $m$, there exists a 
$3$-connected $13$-crossing-critical graph $G(d,m)$, 
which contains at least $m$ vertices of degree at least~$d$.
\end{theorem}

\begin{corollary}
\label{cr:construction}
For every positive integers $c\ge 13$, $d$ and $m$, there exists a 
$3$-connected $c$-crossing-critical graph $G(c,d,m)$, 
which contains at least $m$ vertices of degree at least~$d$.
\end{corollary}

%
%


\medskip
The paper is structured as follows. The preliminaries, needed to 
help understanding the structure of large $c$-crossing critical graphs
are defined in Section \ref{sec:drawing-crossing}.
We prove Theorem \ref{th:boundedMD} in Section \ref{sc:boundedMD},
and Theorem \ref{th:construction13} in Section \ref{sc:construction}.
An additional technical treatment and an operation call zip product is needed
to establish Corollary~\ref{cr:construction} in Section \ref{sc:extended}.  
We conclude with some remarks and open problems in Section \ref{sc:conclusion}.

\section{Graphs and the crossing number}
\label{sec:drawing-crossing}

In this paper, we consider multigraphs by default, even though we could always subdivide
parallel edges (while sacrificing $3$-connectivity) in order to make
our graphs simple.
We follow basic terminology of topological graph theory, see e.g.~\cite{moharThomassen01}.

A {\em drawing} of a graph $G$ in the plane is such that 
the vertices of $G$ are distinct points 
and the edges are simple (polygonal) curves joining their end vertices.
It is required that no edge passes through a vertex, 
and no three edges cross in a common point.
A {\em crossing} is then an intersection point of two edges other than
their common end.  A \emph{face} of the drawing is a maximal connected
subset of the plane minus the drawing.
A drawing without crossings in the plane is called a {\em plane drawing}
of a graph, or shortly a {\em plane graph}.
A graph having a plane drawing is {\em planar}.

The following are the core definitions used in this work.
\begin{definition}[crossing number]
\label{def:crossingn}
The {\em crossing number} $\crn(G)$ of a graph $G$
is the minimum number of crossings of edges in a drawing of $G$ in the plane.
An \emph{optimal drawing} of $G$ is every drawing with exactly $\crn(G)$ crossings.
\end{definition}

\begin{definition}[crossing-critical]
\label{def:crosscritical}
Let $c$ be a positive integer.
A graph $G$ is \emph{$c$-crossing-critical} if $\crn(G)\ge c$, but every proper
subgraph $G'$ of $G$ has $\crn(G')<c$.
\end{definition}
Let us remark that a $c\,$-crossing-critical graph may have no drawing
with only $c$ crossings (for~$c=2$, such an example is the Cartesian product 
of two 3-cycles, $C_3 \Box C_3$).
\doneTodo{Referee 1: 
after def 2.2: $C_3\square C_3$ is not a standard thing, please
give the definition. \\Response: Done as suggested (Drago). }

Suppose $G$ is a graph drawn in the plane with crossings.  Let $G'$ be the
plane graph obtained from this drawing by replacing the crossings with new
vertices of degree $4$.  We say that $G'$ is {the plane graph
associated with the drawing}, shortly the \emph{planarization} of (the drawing of) $G$, 
and the new vertices are the \emph{crossing vertices} of $G'$.


In some of our constructions, we will have to combine crossing-critical graphs as 
described in the next definition.

\begin{definition}
\label{def:zip}
Let $d=2$ or $3$. For $i\in\{1,2\}$, let $G_i$ be a graph and let $v_i\in V(G_i)$
be a vertex of degree $d$ 
\doneTodo{Referee 1: 
Def 2.3: Instead of the somewhat technical $d\in\{ 2, 3\}$,
start with "Let $d=2$ or $3$." \\Response: Done as suggested (Drago). }
that is only incident with simple edges,
such that $G_i-v_i$ is connected.
Let $u_i^j$, $j\in\{1,\ldots,d\}$ be the neighbors of $v_i$. 
The \textsl{zip product} of $G_1$ and $G_2$ at $v_1$ and $v_2$ is obtained
from the disjoint union of $G_1-v_1$ and $G_2-v_2$ by adding the edges
$u_1^ju_2^j$, for each $j\in\{1,\ldots,d\}$. 
\end{definition}
Note that, for different labellings of the neighbors of $v_1$ and $v_2$, 
different graphs may result from the zip product. However, the following
has been shown:
\begin{theorem}[\cite{bokalChimaniLeanos13}]
\label{thm:zip}
Let $G$ be a zip product of $G_1$ and $G_2$ as in Definition \ref{def:zip}.
Then, $\crn(G)=\crn(G_1)+\crn(G_2)$. Furthermore, if for both $i=1$ and $i=2$,
$G_i$ is $c_i$-crossing-critical, where $c_i=\crn(G_i)$, then $G$ is $(c_1+c_2)$-crossing-critical.  
\end{theorem}
\doneTodo{Referee 1: 
Theorem 2.4: I don't understand "Consequently".
Just from the fact, that $cr(G)=cr(G_1)+cr(G_2)$, and
both graphs are crossing critical, why G is $(c_1+c_2)$-crossing critical?
\\Response: Replaced ``Consequently" with ``Furthermore". The proof is not
difficult, but not trivial either (Drago). }

For vertices of degree $2$, this theorem was established 
already by Lea\~{n}os and Salazar in~\cite{leanosSalazar08}.


\section{Structure of $\boldsymbol c$-crossing-critical graphs with large maximum degree}
\label{sc:structure}

Dvo\v{r}\'ak, Hlin\v{e}n\'{y}, and Mohar~\cite{dvorakHlinenyMohar18} recently characterized 
the structure of large $c$-crossing-critical graphs.
From their result, it can be derived that in a crossing-critical graph with a vertex of large degree,
there exist many internally vertex-disjoint paths from this vertex to the boundary of a single face.  To keep our contribution self-contained, we give a simple
independent proof.
We are going to apply this structural result to exclude the existence of large degree vertices
in $c$-crossing-critical graphs for $c\le 12$.

Structural properties of crossing-critical graphs have been studied for more
than two decades, and we now briefly review 
some of the previous important results which we shall use.

Richter and Thomassen~\cite{richterThomassen93} proved the following upper bound:
\begin{theorem}[\cite{richterThomassen93}]
	\label{thm:Richter-Thom}
	Every $c$-crossing-critical
	graph has a drawing with at most $\lceil 5c/2+16\rceil$ crossings.
\end{theorem}

Hlin\v{e}n\'y~\cite{hlineny03} proved that $c\,$-crossing-critical
graphs have path-width bounded in terms of~$c$. 
\begin{theorem}[\cite{hlineny03}]
	\label{thm:bounded-pw}
	There exists a function $f_{\ref{thm:bounded-pw}}:\mathbb{N}\to\mathbb{N}$ such that,
	for every integer $c\ge 1$, every $c$-crossing-critical graph has path-width
	fewer than $f_{\ref{thm:bounded-pw}}(c)$.
\end{theorem}

For simplicity, we omit the exact definition of path-width; 
rather, we only use the following fact~\cite{bienstockRobertsonSeymourThomas91}.
For a rooted tree $T$, let $b(T)$ denote the maximum depth of a rooted complete
binary tree which appears in $T$ as a rooted minor (the \emph{depth} of a rooted tree
is the maximum number of edges of a root-leaf path).
\begin{lemma}\label{lemma:pwobstr}
	For every integer $p\ge 0$, if a graph $G$ either
	\begin{itemize}
		\item contains a subtree $T$ which can be rooted so that $b(T)\ge p$, or
		\item contains pairwise vertex-disjoint paths $P_1$, \ldots, $P_p$ and pairwise vertex-disjoint paths $Q_1$, \ldots, $Q_p$ such that
		$P_i$ intersects $Q_j$ for every $i,j\in\{1,\ldots,p\}$,
	\end{itemize}
	then $G$ has path-width at least $p$.
\end{lemma}

Hlin\v{e}n\'{y} and Salazar~\cite{hlinenySalazar10} also proved that distinct vertices 
in a crossing-critical graph cannot be joined by too many paths.
\begin{theorem}[\cite{hlinenySalazar10}]\label{thm:boundedBond}
	There exists a function $f_{\ref{thm:boundedBond}}:\mathbb{N}\to\mathbb{N}$ such that,
	for every integer $c\ge 1$, no two vertices of a $c$-crossing-critical graph are
	joined by more than $f_{\ref{thm:boundedBond}}(c)$ internally vertex-disjoint paths.
\end{theorem}

As seen in the construction of Dvo\v{r}\'ak and Mohar~\cite{dvorakMohar10} and in the construction
we give in Section~\ref{sc:construction}, crossing-critical graphs can contain arbitrarily many cycles intersecting in exactly one
vertex.  However, such cycles cannot be drawn in a nested way.
A \emph{$1$-nest} of depth $m$ in a plane graph $G$ is a sequence $C_1$, \ldots, $C_m$ of cycles in $G$ 
and a vertex $w\in V(G)$ such that, for $1\le i<j\le m$, the cycle $C_i$ is drawn in the closed disk bounded by $C_j$ and
$V(C_i)\cap V(C_j)=\{w\}$ (Figure~\ref{fig:nests}).  
Hern\'{a}ndez-V\'{e}lez et al.~\cite{velezSalazarThomas12} have shown the following.
\begin{theorem}[\cite{velezSalazarThomas12}]\label{thm:boundedNest}
	There exists a function $f_{\ref{thm:boundedNest}}:\mathbb{N}\to\mathbb{N}$ such that,
	for every integer $c\ge 1$, the planarization of every optimal drawing of a $c$-crossing-critical
	graph does not contain a $1$-nest of depth $f_{\ref{thm:boundedNest}}(c)$.
\end{theorem}

\begin{figure}[tb]
	\centering
	\begin{tikzpicture}[scale=0.7]
	\tikzstyle{every node}=[draw, shape=circle, minimum size=4pt,inner sep=1pt, fill=white]
	\tikzstyle{every path}=[color=black]
	\draw (0.3,0) ellipse (33mm and 21mm);
	\draw (0.02,0) ellipse (30mm and 18mm);
	\draw (-0.26,0) ellipse (27mm and 15mm);
	\draw (-0.55,0) ellipse (24mm and 12mm);
	\draw (-0.84,0) ellipse (21mm and 9mm);
	\draw (-1.31,0) ellipse (16mm and 6mm);
	\node[label=left:$w$] (n1) at (-2.95,0) {};
	\end{tikzpicture}
	\caption{An illustration of a $1$-nest.}
	\label{fig:nests}
\end{figure}

The key structure we use in the proof of Corollary~\ref{cr:haveFan} is a \emph{fan-grid}, which is defined as follows:
\begin{definition}
	Let $G$ be a plane graph and let $v$ be a vertex incident with the outer face of $G$.
	Let $C$ be a cycle in $G$, and let the path $C-v$ be the concatenation of
	vertex-disjoint paths $L$, $Q_1$, \ldots, $Q_n$, $R$ in that order.  Let $H$ be the subgraph
	of $G$ drawn inside the closed disk bounded by $C$.  We say that
	$(v,C,Q_1,\ldots,Q_n)$ is an \emph{$(r\times n)$-fan-grid} with \emph{center} $v$ if
	\begin{itemize}
		\item $H$ contains $n$ internally vertex-disjoint paths $P_1$, \ldots, $P_n$ (the \emph{rays} of the fan-grid), where $P_i$ joins
		$v$ with a vertex of $Q_i$ for $i=1,\ldots, n$, and
		\item $H-V(Q_1\cup \ldots\cup Q_n)$ contains $r$ vertex-disjoint paths from $V(L)$ to $V(R)$. See Figure~\ref{fig:fan-grid}.
	\end{itemize}
\end{definition}

\begin{figure}
	\centering
	\begin{tikzpicture}[scale=0.6]
	\tikzset{every node/.style={simple}}
	\tikzset{every path/.style={thick}}
	\path[use as bounding box] (-6,0) rectangle (6,10);
	
	\node[labeled] (v) at (0,1) {$v$};
	\node[simple] (a1) at (-6,8) {};
	\node[simple] (a2) at (-2,9) {};
	\node[simple] (a3) at (2,9) {};
	\node[simple] (a4) at (6,8) {};
	
	\node[simple] (b1) at (-5,6) {};
	\node[simple] (b2) at (-2.5,6) {};
	\node[simple] (b3) at (0.2,7) {};
	\node[simple] (b4) at (2.5,6) {};
	\node[simple] (b5) at (4,6) {};
	\node[simple] (b6) at (-1.6,4.5) {};
	\node[simple] (b7) at (0.2,5.5) {};
	
	\node[simple] (c1) at (-2.4,3) {};
	\node[simple] (c2) at (-1,3.5) {};
	\node[simple] (c3) at (0,2.5) {};
	\node[simple] (c4) at (1,3) {};
	\node[simple] (c5) at (2.4,3) {};
	
	\node[simple] (d1) at (-3,4) {};
	\node[simple] (d2) at (0.2,4.3) {};
	\node[simple] (d3) at (-1,6) {};
	
	\node[fill=black!16!white] (p1) at (-4.2,8.45) {};
	\node[fill=black!16!white] (p2) at (0.2,8.6) {};
	\node[fill=black!16!white] (p3) at (3.8,8.8) {};
	
	\begin{scope}[on background layer]
	\draw[darkgreen!50!white,line width=4pt] (b1.center) -- (b2.center) -- (b6.center) -- (d3.center) -- (b7.center) -- (b3.center) -- (b4.center) -- (b5.center);
	\draw[darkgreen!50!white,line width=4pt] (c1.center) -- (c2.center) -- (d2.center) -- (c4.center) -- (c5.center);
	
	\end{scope} 
	\draw[red] (v)-- (c1)--(d1)--(b1)--(a1)--(p1)--(a2)--(p2)--(a3)--(p3)--(a4)--(b5)--(c5)--(v);
	\draw[blue] (v)-- (c2)--(b6) -- (b2)--(p1);
	\draw[blue] (v)-- (c3)--(d2)-- (b7) -- (b3)--(p2);
	\draw[blue] (v)-- (c4)--(b4)--(p3);
	
	\coordinate [label=left:\textcolor{black}{$L$}] (L) at (-4.3,5); 
	\coordinate [label=left:\textcolor{black}{$Q_1$}] (Q1) at (-4.3,8.9); 
	\coordinate [label=left:\textcolor{black}{$Q_2$}] (Q2) at (-0.5,9.2); 
	\coordinate [label=left:\textcolor{black}{$Q_3$}] (Q3) at (3.5,9.4); 
	\coordinate [label=left:\textcolor{black}{$R$}] (R) at (4,4.5); 
	\coordinate [label=left:\textcolor{black}{$P_1$}] (P1) at (-2.6,7.5); 
	\coordinate [label=left:\textcolor{black}{$P_2$}] (P2) at (1,7.8); 
	\coordinate [label=left:\textcolor{black}{$P_3$}] (P3) at (3.8,7); 
	
	\end{tikzpicture} 
	
	\caption{ A $(2\times 3)$-fan-grid with center $v$. The rays of this 
		fan-grid ($P_1,P_2,$ and $P_3$) are colored blue. The underlying cycle
		$C$ is red, and the two vertex-disjoint paths from $V(L)$ to $V(R)$
		are green. These paths are shown in the idealized situation where they cross
		each of the paths $P_i$ only once.}
	\label{fig:fan-grid}
\end{figure}

In the argument, we start with a $(0\times n)$-fan-grid and keep enlarging it (adding new rows while sacrificing some of the rays)
as long as possible.  The following definition is useful when looking for the new rows.  A {\em comb with teeth} $v_1,\ldots,v_k$ is a tree consisting of a path $P$
(the \emph{spine} of the comb) and vertex-disjoint paths $P_1,\ldots,P_k$ of length at least one, such that
$P_i$ joins $v_i$ to a vertex in $P$.  We start with simple observations on combs in trees
with many leaves.

\begin{lemma}\label{lemma-boundleaves}
	There exists a function $f_{\ref{lemma-boundleaves}}:\mathbb{N}^3\to\mathbb{N}$ such that the following holds
	for all integers $D,k\ge 1$ and $b\ge 0$.
	Let $T$ be a rooted tree of maximum degree at most $D$ satisfying $b(T)\le b$.
	If every root-leaf path in $T$ contains fewer than $k$ vertices with at least two children,
	then $T$ has at most $f_{\ref{lemma-boundleaves}}(D,b,k)$ leaves.
\end{lemma}
\begin{proof}
	Let $f_{\ref{lemma-boundleaves}}(D,b,k)=1$ if $k=1$ or $b=0$, and 
	$$f_{\ref{lemma-boundleaves}}(D,b,k)=f_{\ref{lemma-boundleaves}}(D,b,k-1)+(D-1)\cdot f_{\ref{lemma-boundleaves}}(D,b-1,k-1)$$
	if $k\ge 2$ and $b\ge 1$.
	We prove the claim by the induction on the number of vertices of $T$.  If $|V(T)|=1$, then
	$T$ has only one leaf.  Hence, suppose that $|V(T)|\ge 2$.  Let $v$ be the root of $T$
	and let $T_1$, \ldots, $T_d$ be the components of $T-v$, where $d=\deg (v)\le D$.
	If $d=1$, then the claim follows by the induction hypothesis applied to $T_1$; hence, suppose that $d\ge 2$.
	In particular, $b\ge b(T)\ge 1$ and $k\ge 2$.
	Then, for all $i\in\{1,\ldots,d\}$, each root-leaf path in $T_i$ contains fewer than $k-1$ vertices with at least
	two children.  Furthermore, there exists at most one $i\in\{1,\ldots,d\}$ such that $b(T_i)=b$;
	hence, we can assume that $b(T_i)\le b-1$ for $2\le i\le d$.  By the induction hypothesis,
	$T_1$ has at most $f_{\ref{lemma-boundleaves}}(D,b,k-1)$ leaves and each of $T_2$, \ldots, $T_d$ has at most $f_{\ref{lemma-boundleaves}}(D,b-1,k-1)$
	leaves, implying the claim.
\end{proof}

\begin{corollary}\label{cor-diameter}
	For every triple of integers satisfying $D,k\ge 1$ and $b\ge 0$,
	every rooted tree $T$ of maximum degree at most $D$, $b(T)\le b$, and with more than $f_{\ref{lemma-boundleaves}}(D, b, k)$ leaves
	contains a comb with $k$ teeth, all of which are leaves in $T$.
\end{corollary}
\begin{proof}
	By Lemma~\ref{lemma-boundleaves}, $T$ contains a root-leaf path $P$ with at least $k$ vertices that have at least two children.
	A subpath of $P$ together with the paths from $k$ of these vertices to leaves forms a comb with $k$ teeth.
\end{proof}

Suppose $Q$ is a path and $K$ is a comb in a plane graph $G$, such that all teeth of $K$ lie on $Q$ and
$K$ and $Q$ are otherwise disjoint.  We say that the comb is \emph{$Q$-clean} if both $Q$ and the spine of $K$
are contained in the boundary of the outer face of the subdrawing of $G$ formed by $K\cup Q$.

\begin{observation}\label{obs-clean}
	Suppose $Q$ is a path and $K$ is a comb in a plane graph $G$, such that all teeth of $K$ lie on $Q$ and
	$K$ and $Q$ are otherwise disjoint.  Let $k\ge 2$ be an integer.  If $Q$ is contained in the boundary of the face of $G$ and
	$K$ has at least $3k-1$ teeth, then $K$ contains a $Q$-clean subcomb with at least $k$ teeth.
\end{observation}

Our aim is to keep growing a fan-grid using the following lemma (increasing $r$ at the expense of sacrificing some of the rays,
see the outcome (d)) until we either obtain a structure that cannot appear in a planarization of a $c$-crossing-critical
graph (outcomes (a)--(c)), or are blocked off from further growth by many rays ending in the boundary of the same face (outcome (e)).

\begin{lemma}\label{lemma-extend}
	There exists a function $f_{\ref{lemma-extend}}:\mathbb{N}^5\to\mathbb{N}$ such that the following holds.
	Let $G$ be a plane graph with a vertex $v$ incident with the outer face.  Let $D$, $b$, $m$, $r$, $k$, and $t$
	be positive integers.  Let $n=f_{\ref{lemma-extend}}(D,b,m,k,t)$. If $G$ contains an $(r\times n)$-fan-grid with center $v$,
	then $G$ also contains at least one of the following substructures:
	\begin{itemize}
		\item[(a)] two vertices joined by more than $D$ internally vertex-disjoint paths, or
		\item[(b)] a $1$-nest of depth greater than $m$, or
		\item[(c)] a subtree $T$ which can be rooted so that $b(T)>b$, or
		\item[(d)] an $((r+1)\times k)$-fan-grid with center $v$, or
		\item[(e)] more than $t$ internally vertex-disjoint paths from $v$ to distinct vertices contained in the boundary of a single face of $G$.
	\end{itemize}
\end{lemma}

\begin{proof}
	Let $s_1=2f_{\ref{lemma-boundleaves}}(D, b, 3k+5)$ and $s_2=2s_1m$.  For an integer $l\ge 0$, let
	$d(l)=(t\cdot(s_1-1))^l$.
	Let $f_{\ref{lemma-extend}}(D,b,m,k,t)=td(s_2)+1$.
	Let $(v,C,Q_1,\ldots,Q_n)$ be an $(r\times n)$-fan-grid in $G$.
	Let $G_1$ be the graph obtained from $G$ by removing the vertices and edges drawn in the open disk bounded by $C$, and
	let $f$ denote the resulting face bounded by $C$.
	
	Suppose that, for some $i\in\{1,\ldots, n\}$, there exists a component $R$ of $G_1-(V(C)\setminus V(Q_i))$ and a set
	$J\subseteq \{1,\ldots,n\}\setminus\{i\}$ of size more than $s_1$ such that $Q_j$
	has a neighbor in $R$ for every $j\in J$.  By symmetry, we can assume that there exists $J'\subseteq J$ of size more than
	$f_{\ref{lemma-boundleaves}}(D, b, 3k+5)$ such that $j>i$ for each $j\in J'$.  Observe that $G$ contains a tree $T$
	with all internal vertices in $R$ and exactly $|J'|$ leaves, one in each of $Q_j$ for $j\in J'$; we root $T$ in a vertex belonging to $R$.
	By Corollary~\ref{cor-diameter}, $\Delta(T)>D$ or $b(T)>b$ or $T$ contains a comb $K$ with $3k+5$ teeth, all of which are
	leaves of $T$.  In the former two cases, $G$ contains (a) or (c).  In the last case, we extract a $(C-v)$-clean subcomb with $k+2$ teeth from $K$
	using Observation~\ref{obs-clean} and combine it with a part
	of the $(r\times n)$-fan-grid in $G$ to form an $((r+1)\times k)$-fan-grid in $G$, showing that $G$ contains (d).
	
	Therefore, we can assume that the following holds:
	\\($\star$) For every $i\in \{1,\ldots, n\}$, every component $R$ of $G_1-(V(C)\setminus V(Q_i))$
	has neighbors in fewer than $s_1$ paths $Q_1$, \ldots, $Q_n$ other than $Q_i$.
	
	A \emph{$C$-bridge} of $G_1$ is either a graph consisting of a single edge of $E(G_1)\setminus E(C)$ and its ends,
	or a graph consisting of a component of $G_1-V(C)$ together with all edges between the component and $C$ and their endpoints.
	For a $C$-bridge $H$, let $J(H)$ denote the set of indices $j\in \{1,\ldots, n\}$ such that $H$ intersects $Q_j$.
	By ($\star$), we have $|J(H)|\le s_1$.  For two $C$-bridges $H_1$ and $H_2$, we write
	$H_1\prec H_2$ if $\min(J(H_2))\le \min(J(H_1))$, $\max(J(H_1))\le \max(J(H_2))$, and either
	at least one of the inequalities is strict or $J(H_2)\subsetneq J(H_1)$ (note that in the last case,
	the planarity implies $|J(H_2)|=2$).
	
	Suppose there exist $s_2+1$ $C$-bridges $H_0,H_1, \ldots, H_{s_2}$ such that $H_0\prec H_1\prec \ldots\prec H_{s_2}$.
	For $1\le j\le m+1$, let $b(j)=2(j-1)s_1$.  If, for $1\le j\le m$, we have
	$\min(J(H_{b(j)}))>\min(J(H_{b(j+1)}))$ and $\max(J(H_{b(j)}))<\max(J(H_{b(j+1)}))$,
	then $G$ contains (b).  Hence, by symmetry we can assume that there exists $j\in \{1,\ldots,m\}$
	such that $\min(J(H_{b(j)}))=\min(J(H_{b(j+1)}))$.  Consequently, there exists an index $i$
	such that $\min(J(H_p))=i$, for  $b(j)\le p\le b(j+1)$.  But then $|\bigcup_{p=b(j)}^{b(j+1)} J(H_p)\setminus\{i\}|>s_1$,
	which contradicts ($\star$).
	
	Consequently, there is no chain of order greater than $s_2$ in the partial ordering $\prec$.
	For a $C$-bridge $H$, let $l(H)$ denote the order of the longest chain of $\prec$ with the maximum element $H$.
	Suppose that $\max(J(H))-\min(J(H))>d(l(H))$, and choose a $C$-bridge $H$ with this property such that $l(H)$ is minimum.
	Since $|J(H)|\le s_1$, there exist two consecutive elements $j_1$ and $j_2$ of $J(H)$ such that $j_2-j_1>t\cdot d(l(H)-1)$.
	If $l(H)=1$, this implies there exists a face $f$ of $G$ such that all paths $Q_j$ with $j_1\le j\le j_2$ contain vertices incident with $f$,
	and $G$ contains (e).  Hence, suppose that $l(H)>1$.
	Let $B$ be the set of bridges $H'\prec H$ such that $j_1\le \min(J(H'))\le \max(J(H'))\le j_2$
	that are maximal in $\prec$ with this property.
	By the minimality of $l(H)$, every bridge $H'\in B$ satisfies $\max(J(H'))-\min(J(H'))\le d(l(H)-1)$.
	Consequently, there are more than $t+1$ indices $j$ such that $j_1\le j\le j_2$ and either $j=\max(J(H'))$ for some $H'\in B$
	or there does not exists any bridge $H'\prec H$ such that $\min(J(H'))\le j\le \max(J(H'))$.
	Observe there exists a face $f$ of $G$ such that, for each such index $j$, the path $Q_j$ contains a vertex incident with $f$.
	Hence, $G$ again contains (e).
	
	Consequently, we can assume that $\max(J(H))-\min(J(H))\le d(s_2)$, for each $C$-bridge $H$.
	Since $n>td(s_2)$, applying an analogous argument to the $C$-bridges that are maximal in $\prec$ yields that
	$G$ contains (e).
\end{proof}

To start up the growing process based on Lemma \ref{lemma-extend}, we need to show that a fan-grid with many rays exists.
\begin{lemma}\label{lemma-start}
	Let $G$ be a $2$-connected plane graph with a vertex $v$ incident with the outer face.
	For every positive integers $D$, $b$, and $k$, if $v$ has degree more than $f_{\ref{lemma-boundleaves}}(D, b+1, 3k+5)$,
	then $G$ contains at least one of the following:
	\begin{itemize}
		\item[(a)] two vertices joined by more than $D$ internally vertex-disjoint paths, or
		\item[(c)] a subtree $T$ which can be rooted so that $b(T)>b$, or
		\item[(d)] a $(0\times k)$-fan-grid with center $v$.
	\end{itemize}
\end{lemma}
\begin{proof}
	Let $G'$ be the graph obtained from $G$ by splitting $v$ into vertices of degree $1$, and let $S$ be the
	set of these vertices.  Since $G$ is $2$-connected, $G'$ is connected, and thus it contains a subtree $T$
	whose leaves coincide with $S$.  Root $T$ arbitrarily in a non-leaf vertex.
	By Corollary~\ref{cor-diameter}, $\Delta(T)>D$ or $b(T)>b+1$ or 
	 $T$ has a comb with $3k+5$ teeth in $S$.  In the first case, (a) holds.  In the second case,
	$b(T-S)>b$ and $T-S$ is a subtree of $G$, and thus (c) holds.  In the last case, we can extract a $v$-clean subcomb with
	at least $k+2$ teeth using Observation~\ref{obs-clean}, which gives rise to a $(0\times k)$-fan-grid with center $v$ in $G$.
\end{proof}

Note that a $((p+1)\times(p+1))$-fan-grid contains two systems of $p+1$ pairwise vertex-disjoint paths such that
every two paths from the two systems intersect; hence, by Lemma~\ref{lemma:pwobstr} a plane graph of path-width
at most $p$ contains neither a $((p+1)\times(p+1))$-fan-grid nor a subtree $T$ which can be rooted so that $b(T)>p$.
Hence, starting from Lemma~\ref{lemma-start} and iterating Lemma~\ref{lemma-extend} at most $p+1$ times, we obtain the following.
\begin{corollary}\label{cr:bondNestFan}
	There exists a function $f_{\ref{cr:bondNestFan}}:\mathbb{N}^4\to\mathbb{N}$ such that the following holds.
	Let $G$ be a $2$-connected plane graph.  Let $D$, $m$, $p$, and $t$ be positive integers.  Let $\Delta=f_{\ref{cr:bondNestFan}}(D,m,p,t)$.  If $G$
	has path-width at most $p$ and maximum degree greater than $\Delta$, 
	then $G$ contains at least one of the following:
	\begin{itemize}
		\item[(a)] two vertices joined by more than $D$ internally vertex-disjoint paths, or
		\item[(b)] a $1$-nest of depth greater than $m$, or
		\item[(e)] more than $t$ internally vertex-disjoint paths from a vertex $v$ to distinct vertices contained in the boundary of a single face of $G$.
	\end{itemize}
\end{corollary}

We now apply this result to an optimal planarization of a $c$-crossing-critical graph.
\begin{corollary}
	\label{cr:haveFan}
	There exists a function $f_{\ref{cr:haveFan}}:\mathbb{N}^2\to\mathbb{N}$ such that the following holds.
	Let $c\ge 1$ and $t\ge 3$ be integers and let $G$ be an optimal drawing of a $2$-connected $c$-crossing-critical graph.
	If $G$ has maximum degree greater than $f_{\ref{cr:haveFan}}(c,t)$, then there exists
	a path $Q$ contained in the boundary of a face of $G$ and internally vertex-disjoint paths $P_1$, \ldots, $P_t$
	starting in the same vertex not in $Q$ and ending in distinct vertices appearing in order on $Q$ (and otherwise disjoint from $Q$),
	such that no crossings of $G$ appear on $P_1$, $P_t$, nor in the face of $P_1\cup P_t\cup Q$ that contains $P_2$, \ldots, $P_{t-1}$.
\end{corollary}
\doneTodo{Referee 4: 
	line 140: I find the part after "such that" misleading/wrong. In particular, I think that "or" should be "and".\\
	Response: Done. "or" was replaced by "nor". (Drago)}
\doneTodo{Referee 1: 
	Corollary 2.5: why is our function denoted by $f_{2.5}$? very strange notation.
	Also here: a figure would be very useful. (at least in the full version) 
	\\Response: In the proof of this corollary, there is a series of similar statements,
	each with its own function. Hence, we denoted the function of each claim with the label
	of that claim. As this is the function of Corollary 2.5, its denotation is $f_{2.5}$.
	For the short statement, it would make sense to simplify it, but as this function does
	not appear anywhere and we would like to refrain from maintaining several copies of the 
	LaTeX code, we prefer to keep it as it is. As suggested, there are more figures planned, 
	and they will be added at the latest to the journal version. (Drago). }

\begin{proof}
	Let $c'=\lceil 5c/2+16\rceil$, $D=\max(5,f_{\ref{thm:boundedBond}}(c)+c')$, $m=f_{\ref{thm:boundedNest}}(c)$, $p=f_{\ref{thm:bounded-pw}}(c)+c'$,
	and $f_{\ref{cr:haveFan}}(c,t)=f_{\ref{cr:bondNestFan}}(D,m,p,(c'+2)t)$.
	
	By Theorem~\ref{thm:Richter-Thom}, $G$ has at most $c'$ crossings.  Let $G'$ be the planarization of $G$.
	Note that $G'$ is $2$-connected, since otherwise a crossing vertex would form a cut in $G'$ and the corresponding crossing
	in $G$ could be eliminated, contradicting the optimality of the drawing of $G$.  By Theorem~\ref{thm:bounded-pw},
	$G$ has path-width at most $p-c'$, and thus $G'$ has path-width at most $p$.  By Theorem~\ref{thm:boundedBond},
	$G$ does not contain more than $D-c'$ internally vertex-disjoint paths between any two vertices, and thus $G'$
	does not contain more than $D$ internally vertex-disjoint paths between any two vertices.  By Theorem~\ref{thm:boundedNest},
	$G'$ does not contain a $1$-nest of depth $m$.  Hence, by Corollary~\ref{cr:bondNestFan}, $G'$ contains
	more than $(c'+1)t$ internally disjoint paths from a vertex $v$ to distinct vertices contained in the boundary of a single face $f$ of $G'$.
	Let $Q_1$, \ldots, $Q_{c'+2}$ be disjoint paths contained in the boundary of $f$ such that, for $i=1,\ldots, c'+2$,
	$t$ of the paths $P_{i,1}$, \ldots, $P_{i,t}$ from $v$ end in $Q_i$ in order.  Let $g_i$ denote the face of $Q_i\cup P_{i,1}\cup P_{i,t}$
	containing $P_{i,2}$, \ldots, $P_{i,t-1}$.  Note that the closures of $g_1$, \ldots, $g_{c'+2}$ intersect only in $v$ 
	and since $G'$ contains at most $c'$ crossing vertices, there exists $i\in\{1,\ldots, c'+2\}$ such that no crossing vertex
	is contained in the closure of $g_i$ and $v$ is not in $Q_i$.  Hence, for $j=1,\ldots, t$, we can set $Q=Q_i$ and $P_j=P_{i,j}$.
\end{proof}

\section{Crossing-critical graphs with at most 12 crossings}
\label{sc:boundedMD} 

We now use Corollary~\ref{cr:haveFan} to prove the following ``redrawing'' lemma.

\begin{lemma}\label{lemma-redraw}
Let $G$ be a $2$-connected $c$-crossing-critical graph.  If $G$ has maximum degree greater than $f_{\ref{cr:haveFan}}(c,6c+1)$,
then there exist integers $r\ge 2$ and $k\ge 0$ such that $kr\le c-1$ and $G$ has a drawing with at most $c-1-kr+\binom{k}{2}$ crossings.
\end{lemma}

\begin{figure}[p]
 \begin{tikzpicture}[xscale=1.15, yscale=0.7]\small
  \tikzset{every node/.style={labeled}}
  \tikzset{every path/.style={thick}}

\node (v) at (0,0) {$v$};
\node[simple] (p0) at (-3,3){};
\node[simple] (p1) at (-2,4.1){};
\node[simple] (p2) at (0,4.7){};
\node[simple] (p3) at (2,4.1){};
\node[simple] (p4) at (3,3){};
\node[simple] (p5) at (3.5,1){};
\node[simple] (p51) at (1.16,0.33){};
\node[simple] (p6) at (3.2,-1.5){};
\node[simple] (p61) at (1.06,-0.5){};
\node[simple] (p62) at (1.6,-0.75){};
\node[simple] (p63) at (2.4,-1.15){};
\node[simple] (p7) at (1.3,-3.5){};
\node[simple] (p71) at (0.52,-1.45){};
\node[simple] (p72) at (0.8,-2.2){};
\node[simple] (p8) at (-1.3,-3.5){};
\node[simple] (p81) at (-0.52,-1.45){};
\node[simple] (p82) at (-0.8,-2.2){};
\node[simple] (p9) at (-3.2,-1.5){};
\node[simple] (p91) at (-1.06,-0.5){};
\node[simple] (p92) at (-1.6,-0.75){};
\node[simple] (p93) at (-2.4,-1.15){};
\node[simple] (p10) at (-3.5,1){};
\node[simple] (p101) at (-1.16,0.33){};

\tikzset{every node/.style={edge-label}}
 
\draw (v) -- (p0) -- (p1)-- (p2) -- (p3)-- (p4) --  (v);
\draw (p4) -- (p5); 
\draw (p10) -- (p0);
\draw (v) -- (p1); 
\draw (v) -- (p3);
 \draw[black,dotted] (0,1.3)--(p2);
 \draw (v)-- (-0.25,1);
 \draw (v)-- (0.25,1);
 \draw (v)-- (0,1);
 \draw (-0.5,1) .. controls (-0.25,1.4) and (0.25,1.4) ..  (0.5,1);
\draw[blue] (v) -- (p101) -- (p10);
\draw[blue] (v) -- (p91) -- (p92) -- (p93) -- (p9);
\draw[blue] (v) -- (p51) -- (p5);
\draw[blue] (v) -- (p61) -- (p62)  -- (p63) -- (p6); 
\draw[red] (p62) -- (p71) -- (p81) --  (p92);
\draw[red] (p63) -- (p72) -- (p82) --  (p93);

\begin{scope}[on background layer]	
	\draw[blue!50!white,line width=1pt] (p5.center) -- (p6.center);
	\draw[blue!50!white,line width=1pt] (p9.center) -- (p10.center);
	\draw[red!100!white,line width=1pt] (v.center)  to[bend right=14] (p71.center)  to[bend right=14] (p72.center)  to[bend right=14] (p7.center);
	\draw[red!100!white,line width=1pt] (v.center)  to[bend left=14] (p81.center)  to[bend  left=14] (p82.center)  to[bend  left=14] (p8.center);
	\draw[red!100!white,line width=1pt] (v.center)  to[bend left=14] (p91.center)  to[bend  left=14] (p92.center)  to[bend  left=14] (p93.center) to[bend left=12]  (p9.center);
	\draw[red!70!white,line width=1pt] (p6.center) -- (p7.center) -- (p8.center) -- (p9.center);	
	\draw[red!100!white,line width=1pt] (v.center)  to[bend right=14] (p61.center) to[bend right=20]  (p62.center) to[bend right=14]  (p63.center) to[bend right=14]  (p6.center);
	 \end{scope}

\coordinate [label=left:\textcolor{black}{$P_1$}] (P1) at (-0.7,2.4); 
\coordinate [label=left:\textcolor{black}{$P_{6c+1}$}] (P6) at (2,2.4);
\coordinate [label=left:\textcolor{blue}{$S_1$}] (S1) at (-2.2,4.1);
\coordinate [label=left:\textcolor{blue}{$S_2$}] (S2) at (2.8,4);
\coordinate [label=left:\textcolor{blue}{$S_L$}] (S_L) at (-2.4,1.1);
\coordinate [label=left:\textcolor{blue}{$e_{S_L}$}] (e_L) at (-0.45,0.43);
\coordinate [label=left:\textcolor{blue}{$M_L$}] (M_L) at (-1.8,-0.6);
\coordinate [label=left:\textcolor{blue}{$e_{M_L}$}] (e_M) at (-0.5,-0.1);
\coordinate [label=left:\textcolor{blue}{$C_L$}] (C_L) at (-3.4,-0.3);
\coordinate [label=left:\textcolor{blue}{$S_R$}] (S_R) at (2.4,1);
\coordinate [label=left:\textcolor{blue}{$e_{S_R}$}] (e_R) at (1.1,0.43);
\coordinate [label=left:\textcolor{blue}{$M_R$}] (M_R) at (2.4,-0.7);
\coordinate [label=left:\textcolor{blue}{$e_{M_R}$}] (e_R) at (1,-0.1);
\coordinate [label=left:\textcolor{blue}{$C_R$}] (C_R) at (4,-0.3);
\coordinate [label=left:\textcolor{red}{$1$}] (1) at (0.1,-1.2);
\coordinate [label=left:\textcolor{red}{$2$}] (2) at (0.1,-1.94);
\coordinate [label=left:\textcolor{red}{$r$}] (r) at (0.1,-3.2);
\coordinate [label=left:\textcolor{red}{$M_1$}] (M_1) at (-2.3,-2.6);
\coordinate [label=left:\textcolor{red}{$M_2$}] (M_2) at (2.7,-2.8);
\coordinate [label=left:\textcolor{red}{$M$}] (M_1) at (0.1,-3.8);
\coordinate [label=left:\textcolor{black}{a)}] (a) at (-3.48,4);
\end{tikzpicture}


\hfill
 \begin{tikzpicture}[xscale=1.15, yscale=0.7]\small

  \tikzset{every node/.style={labeled}}
  \tikzset{every path/.style={thick}}
\path[use as bounding box] (-5,-5) rectangle (4.75,5);

\node (v) at (0,0) {$v$};
\node[simple] (p0) at (-3,3){};
\node[simple] (p1) at (-2,4.1){};
\node[simple] (p01) at (-2.6,3.5){};
\node[simple] (p2) at (0,4.7){};
\node[simple] (p3) at (2,4.1){};
\node[simple] (p4) at (3,3){};
\node[simple] (p5) at (3.5,1){};
\node[simple] (p51) at (1.16,0.33){};
\node[simple] (p6) at (3.2,-1.5){};
\node[simple] (p61) at (1.06,-0.5){};
\node[simple] (p62) at (1.6,-0.75){};
\node[simple] (p63) at (2.4,-1.15){};
\node[simple] (p7) at (1.3,-3.5){};
\node[simple] (p71) at (0.52,-1.45){};
\node[simple] (p72) at (0.8,-2.2){};
\node[simple] (p73) at (1.05,-2.8){};
\node[simple] (p8) at (-1.3,-3.5){};
\node[simple] (p81) at (-0.52,-1.45){};
\node[simple] (p82) at (-0.8,-2.2){};
\node[simple] (p83) at (-1.05,-2.8){};
\node[simple] (p9) at (-3.2,-1.5){};
\node[simple] (p91) at (-1.06,-0.5){};
\node[simple] (p92) at (-1.6,-0.75){};
\node[simple] (p93) at (-2.4,-1.15){};
\node[simple] (p10) at (-3.5,1){};
\node[simple] (p101) at (-1.16,0.33){};

\tikzset{every node/.style={edge-label}}
 
\draw (v) -- (p0) -- (p01) -- (p1)-- (p2) -- (p3)-- (p4) --  (v);
\draw (p4) -- (p5); 
\draw (p10) -- (p0); 
\draw (v) -- (p1); 
\draw (v) -- (p3);
 \draw[black,dotted] (0,1.3)--(p2);
 \draw (v)-- (-0.25,1);
 \draw (v)-- (0.25,1);
 \draw (v)-- (0,1);
 \draw (-0.5,1) .. controls (-0.25,1.4) and (0.25,1.4) ..  (0.5,1);
\draw[blue] (v) -- (p101) -- (p10);
\draw[blue] (v) -- (p91) -- (p92) -- (p93) -- (p9)--(p10);
\draw[blue] (v) -- (p51) -- (p5);
\draw[blue] (v) -- (p61) -- (p62)  -- (p63) -- (p6) -- (p5); 
\draw[red] (p71) -- (p81);
\draw[red] (p72) -- (p82);
\draw[red] (p73) -- (p83);

\begin{scope}[on background layer]	
	\draw[red!100!white,line width=1pt] (v.center)  to[bend right=14] (p71.center)  to[bend right=14] (p72.center)  to[bend right=14] (p73.center)  
	to[bend right=14] (p7.center);
	\draw[red!100!white,line width=1pt] (v.center)  to[bend left=14] (p81.center)  to[bend  left=14] (p82.center)  to[bend  left=14] (p83.center) to[bend left=14] (p8.center);
	\draw[red!100!white,line width=1pt] (v.center)  to[bend left=14] (p91.center)  to[bend  left=14] (p92.center)  to[bend  left=14] (p93.center) to[bend left=12]  (p9.center);
	\draw[red!100!white,line width=1pt] (p6.center) -- (p7.center) -- (p8.center) -- (p9.center);	
	\draw[red!100!white,line width=1pt] (v.center)  to[bend right=14] (p61.center) to[bend right=20]  (p62.center) to[bend right=14]  (p63.center) to[bend right=14]  (p6.center);
	 \end{scope}
	 
 \draw[darkgreen] (v) .. controls (-0.7,-7) and (-4,-3) ..  (-4.5,0);
  \draw[darkgreen] (-4.5,0) .. controls (-4.7,1.5) and (-4,2) .. (p0);
 \draw[darkgreen] (v) .. controls (0,-8) and (-4.2,-3) ..  (-4.7,0);
 \draw[darkgreen] (-4.7,0) .. controls (-4.9,1.5) and (-4.2,3) .. (p01);
 \draw[darkgreen] (v) .. controls (0.5,-7.3) and (4.2,-3) ..  (4.7,0);
 \draw[darkgreen] (4.7,0) .. controls (4.9,1.5) and (4.2,3) .. (p3);
	 
\coordinate [label=left:\textcolor{black}{$P_1$}] (P1) at (-0.7,2.4); 
\coordinate [label=left:\textcolor{black}{$P_{6c+1}$}] (P6) at (2,2.4);

\coordinate [label=left:\textcolor{blue}{$S_1$}] (S1) at (-2.2,4.2);
\coordinate [label=left:\textcolor{blue}{$S_2$}] (S2) at (3,4.2);

\coordinate [label=left:\textcolor{blue}{$S_L$}] (S_L) at (-2.4,1.1);
\coordinate [label=left:\textcolor{blue}{$e_{S_L}$}] (e_L) at (-0.45,0.43);
\coordinate [label=left:\textcolor{blue}{$M_L$}] (M_L) at (-1.8,-0.6);
\coordinate [label=left:\textcolor{blue}{$e_{M_L}$}] (e_M) at (-0.5,-0.1);
\coordinate [label=left:\textcolor{blue}{$C_L$}] (C_L) at (-3.4,-0.3);
\coordinate [label=left:\textcolor{blue}{$S_R$}] (S_R) at (2.4,1);
\coordinate [label=left:\textcolor{blue}{$e_{S_R}$}] (e_R) at (1.1,0.43);
\coordinate [label=left:\textcolor{blue}{$M_R$}] (M_R) at (2.4,-0.7);
\coordinate [label=left:\textcolor{blue}{$e_{M_R}$}] (e_R) at (1,-0.1);
\coordinate [label=left:\textcolor{blue}{$C_R$}] (C_R) at (4,-0.3);
\coordinate [label=left:\textcolor{red}{$1$}] (C1) at (0.4,-1.25);
\coordinate [label=left:\textcolor{red}{$2$}] (C2) at (0.6,-1.97);
\coordinate [label=left:\textcolor{red}{$3$}] (C3) at (0.8,-2.6);
\coordinate [label=left:\textcolor{red}{$r$}] (C4) at (1,-3.3);
\coordinate [label=left:\textcolor{darkgreen}{$1$}] (e1) at (-1.7,-3.7);
\coordinate [label=left:\textcolor{darkgreen}{$2$}] (e2) at (-1.4,-4.2);
\coordinate [label=left:\textcolor{darkgreen}{$3$}] (e3) at (2.4,-3.7);
\coordinate [label=left:\textcolor{red}{$M_1$}] (M_1) at (-1.3,-1.6);
\coordinate [label=left:\textcolor{red}{$M_2$}] (M_2) at (1.7,-1.6);
\coordinate [label=left:\textcolor{darkgreen}{($k=3$)}] (k) at (0.51,-4.4);
\coordinate [label=left:\textcolor{red}{($r=4$)}] (r) at (0.5,-3.8);
\coordinate [label=left:\textcolor{black}{b)}] (b) at (-4.4,3.2);
\end{tikzpicture}


\vspace*{2ex}
 \begin{tikzpicture}[xscale=1.2, yscale=0.6]\small
  \tikzset{every node/.style={labeled}}
  \tikzset{every path/.style={thick}}
\path[use as bounding box] (-5,-5) rectangle (5,5);

\node (v) at (0,0) {$v$};
\node[simple] (p0) at (-3,3){};
\node[simple] (p1) at (-2,4.1){};
\node[simple] (p11) at (-1.45,3){};
\node[simple] (p3) at (2,4.1){};
\node[simple] (p31) at (0.63,1.3){};
\node[simple] (p32) at (0.9,1.9){};
\node[simple] (p33) at (1.22,2.5){};
\node[simple] (p34) at (1.55,3.2){};
\node[simple] (p4) at (3.4,1.8){};
\node[simple] (p41) at (1,0.5){};
\node[simple] (p5) at (3.5,-0.3){};
\node[simple] (p6) at (2.8,-2.6){};
\node[simple] (p63) at (1,-0.9){};
\node[simple] (p9) at (-3.2,-1.5){};
\node[simple] (p91) at (-1.06,-0.5){};
\node[simple] (p92) at (-1.6,-0.75){};
\node[simple] (p93) at (-2.2,-1.05){};
\node[simple] (p94) at (-2.7,-1.26){};
\node[simple] (p10) at (-3.5,1){};
\node[simple] (p101) at (-1.16,0.33){};

\tikzset{every node/.style={edge-label}}
 
\draw (v) -- (p0)  -- (p1);
\draw (p4) -- (p5); 
\draw (p10) -- (p0); 
\draw (v) -- (p11)--(p1); 
\draw (v) -- (p5);
\draw[blue] (v) -- (p101) -- (p10);
\draw[blue] (v) -- (p91) -- (p92) -- (p93) -- (p94)--(p9)--(p10);
\draw[blue] (v) -- (p31) -- (p32)  -- (p33) -- (p34) -- (p3)--(p4)--(p41)--(v); 
\draw (v) -- (p63) -- (p6)--(p5); 
	 
 \draw (v) .. controls (-4,-4) and (-4,-2) ..  (-4,-1);
 \draw (-4,-1) .. controls (-4,0) and (-4.3,3) .. (p1);
 \draw (p63) .. controls (-5.5,-5) and (-4.5,-2) ..  (-4.2,2);
 \draw (-4.2,2) .. controls (-4.2,4) and (-2,6.5) .. (p11);
 \draw (p6) .. controls (-6.5,-5) and (-5,-2) ..  (-4.5,3);
 \draw (-4.5,3) .. controls (-4,6) and (-1,7.7) .. (v);

 \draw[red] (p9) .. controls (4,-6) and (4,-4) ..  (4,-1);
 \draw[red] (4,-1) .. controls (4,0) and (4.3,3) .. (p3);
 \draw[red] (p94) .. controls (4.7,-5.8) and (4.5,-4.8) ..  (4.5,1);
 \draw[red] (4.5,1) .. controls (4.4,4) and (1.5,6) .. (p34);
\draw[red] (p93) .. controls (5.4,-5.7) and (5.1,-4.7) ..  (4.9,1);
\draw[red] (4.9,1) .. controls (4.7,4) and (1.5,7) .. (p33);
\draw[red] (p92) .. controls (6,-5.4) and (5.6,-4.3) ..  (5.3,1);
\draw[red] (5.3,1) .. controls (5.2,4) and (1.2,8) .. (p32);

\coordinate [label=left:\textcolor{black}{$P_1$}] (P1) at (-0.7,2.4); 
\coordinate [label=left:\textcolor{black}{$P_{6c+1}$}] (P6) at (2.1,-1.1);
\coordinate [label=left:\textcolor{blue}{$S_L$}] (S_L) at (-2.4,1.1);
\coordinate [label=left:\textcolor{blue}{$e_{S_L}$}] (e_L) at (-0.45,0.43);
\coordinate [label=left:\textcolor{blue}{$M_L$}] (M_L) at (-1.8,-0.6);
\coordinate [label=left:\textcolor{blue}{$e_{M_L}$}] (e_M) at (-0.5,-0.1);
\coordinate [label=left:\textcolor{blue}{$C_L$}] (C_L) at (-3.4,-0.3);
\coordinate [label=left:\textcolor{blue}{$S_R$}] (S_R) at (2.5,1.7);
\coordinate [label=left:\textcolor{blue}{$e_{S_R}$}] (e_R) at (1.2,0.15);
\coordinate [label=left:\textcolor{blue}{$M_R$}] (M_R) at (2,2.5);
\coordinate [label=left:\textcolor{blue}{$e_{M_R}$}] (e_R) at (1.15,0.9);
\coordinate [label=left:\textcolor{blue}{$C_R$}] (C_R) at (3.3,3);
\coordinate [label=left:\textcolor{red}{$1$}] (r1) at (4.3,-1.5);
\coordinate [label=left:\textcolor{red}{$2$}] (r2) at (4.79,-1);
\coordinate [label=left:\textcolor{red}{$3$}] (r3) at (5.25,-0.5);
\coordinate [label=left:\textcolor{red}{$r$}] (rr) at (5.65,0);
\coordinate [label=left:\textcolor{black}{$1$}] (b1) at (-4,-1);
\coordinate [label=left:\textcolor{black}{$2$}] (b2) at (-4.4,-0.5);
\coordinate [label=left:\textcolor{black}{$k$}] (bk) at (-4.85,0);
\coordinate [label=left:\textcolor{black}{c)}] (c) at (-4.7,5);
\end{tikzpicture}

 \caption{An illustration of the proof of Lemma \ref{lemma-redraw}. 
~a) The original optimal
drawing of $G$, with subdrawings of $M_1$ and $M_2$ (red) that will be glued into the drawing of $G_0$ from an
optimal drawing of $G-e$. 
~b) A drawing of $G$ with at most $c-1$ crossings, obtained from $G_0$
(black, blue, green) and $M_1,M_2$ (red). 
~c) A drawing of $G$ with at most $\binom{c-1-kr+{k}}{2}$ crossings, obtained
from $G_0$ (black, blue) and $M_1,M_2$ (red).}
 \label{fig:lemma-redraw}

\end{figure}

\begin{proof}
\doneTodo{Referee 2: 
L147:  You need to introduce v as the common end of the Pi.\\
Response: Done as suggested. DB.}

Consider an optimal drawing of $G$. Let $P_1$, \ldots, $P_{6c+1}$ be paths obtained using Corollary~\ref{cr:haveFan} and $v$ their common end vertex.
For $2\le i\le 6c-1$, let $T_i$ denote the $2$-connected block of $G-((V(P_{i-1})\cup V(P_{i+2}))\setminus \{v\})$
containing $P_i$ and $P_{i+1}$, and let $C_i$ denote the cycle bounding the face of $T_i$ containing $P_{i-1}$.
Note that if $2\le i$ and $i+3\le j\le 6c-1$, then $G-V(T_i\cup T_j)$ has at most three components: one containing $P_{i+2}-v$,
one containing $P_1-v$, and one containing $P_{6c+1}-v$, where the latter two components can be the same.

Let $e$ be the edge of $P_{3c+1}$ incident with $v$ and let $G'$ be an optimal drawing of $G-e$.
Since $G$ is $c$-crossing-critical, $G'$ has at most $c-1$ crossings.  Hence, there exist indices $i_1$ and $i_2$
such that $2\le i_1\le 3c-1$, $3c+2\le i_2\le 6c-1$, and none of the edges of $T_{i_1}$ and $T_{i_2}$ is crossed.
Let us set $L=T_{i_1}$, $C_L=C_{i_1}$, $R=T_{i_2}$, and $C_R=C_{i_2}$. Let $M$, $S_1$, and $S_2$ denote the subgraphs
of $G$ consisting of the components of $G-V(L\cup R)$ containing $P_{3c+1}-v$, $P_1-v$, and $P_{6c+1}-v$, respectively,
together with the edges from these components to the rest of $G$ and their incident vertices (where possibly $S_1=S_2$).
Let $S_L$ and $M_L$ be subpaths of $C_L$ of length at least one intersecting in $v$ such that $V(S_1\cap C_L)\subseteq V(S_L)$ and $V(M\cap C_L)\subseteq V(M_L)$.
Analogously, let $S_R$ and $M_R$ be subpaths of $C_R$ of length at least one intersecting in $v$ such that $V(S_2\cap C_R)\subseteq V(S_R)$ and $V(M\cap C_R)\subseteq V(M_R)$.
See Figure~\ref{fig:lemma-redraw}.
\doneTodo{Referee 4: 
line 160: ... such that $V(S_1 \cap C_L)$ ... I think that you meant to relate vertex sets rather than graphs.\\
Response DB: The suggested formulation is weaker and gives more freedom when multiple edges are involved. As this is the case, we adopt the recommendation as suggested, although essentially they follow from the stronger assumptions. Zdenek, please review this one.}

\doneTodo{Referee 1: 
Page 5, line 163: We can assume without loss of generality ...
\\Response: Done (Drago). }

We can assume without loss of generality (by circle inversion of the plane if
necessary) that neither $C_L$ nor $C_R$ bounds the outer face of $C_L\cup C_R$ in the drawings inherited from $G$ and from $G'$.
Let $e_{M_L}$, $e_{S_L}$, $e_{S_R}$, $e_{M_R}$ be the clockwise cyclic order of the edges of $C_L\cup C_R$ incident with $v$ in
the drawing $G$, where $e_Q\in E(Q)$ for every $Q\in \{M_L,S_L,S_R,M_R\}$.
By the same argument, we can assume that the clockwise cyclic order
of these edges in the drawing of $G'$ is either the same or $e_{M_L}$, $e_{S_L}$, $e_{M_R}$, $e_{S_R}$.

\doneTodo{Referee 4: 
line 172: Please, explain the notion of "rearranging drawing".\\
Response: Done in a footnote. Zdenek, please review this one. }
In $G$, $L$ is drawn in the closed disk bounded by $C_L$, $R$ is drawn in the closed disk bounded by $C_R$, and
$M$, $S_1$, and $S_2$ together with all the edges joining them to $v$ are drawn in the outer face of $C_L\cup C_R$.
Since $C_L$ and $C_R$ are not crossed in the drawing $G'$, we can if necessary rearrange the drawing of $G'$ without creating any new
crossings\footnote{As $G$ is not necessarily $3$-connected, it is possible that some $2$-connected components or some edges
of $L,R$ are drawn in the exterior of the disk bounded by $C_L$, $C_R$. However, these
can be flipped into the interior of $C_L$, $C_R$, and after such \textsl{rearranging}, 
$C_L$, $C_R$ bound the outer face of the drawings of $L$, $R$. Similarly, if $S_1\neq S_2$, either of them could be in the interior of $C_L, C_R$, and we flip them into the exterior, so that the interior of $C_L, C_R$ contains only drawings of $L$, $R$, respecitvely.} 
so that the same holds for the drawings of $L$, $R$, $M$, $S_1$, and $S_2$ in $G'$.  Let $r\ge 1$ denote
the maximum number of pairwise edge-disjoint paths in $M-v$ from $V(M\cap C_L-v)$ to $V(M\cap C_R-v)$.
By Menger's theorem, $M-v$ has disjoint induced subgraphs $M'_1$ and $M'_2$ such that $V(M-v)=V(M'_1)\cup V(M'_2)$,
$V(M\cap C_L-v)\subseteq V(M'_1)$, $V(M\cap C_R-v)\subseteq V(M'_2)$, and $G$ contains exactly $r$ edges with one end in $M'_1$
and the other end in $M'_2$.  For $i\in \{1,2\}$, let $M_i$ be the subgraph of $M$ induced by $V(M'_i)\cup \{v\}$.
Let $F$ be a path in $M-v$ from $V(M\cap C_L-v)$ to $V(M\cap C_R-v)$ that has in the drawing $G'$ the smallest number of
intersections with the edges of $S_1\cup S_2$, and let $k$ denote the number of such intersections.  Let $G_0$ denote the drawing
$G'-(V(M)\setminus V(M_L\cup M_R))$.  Since $M-v$ contains $r$ pairwise edge-disjoint paths from $V(M\cap C_L-v)$ to $V(M\cap C_R-v)$
and each of them crosses $S_1\cup S_2$ at least $k$ times, we conclude that $G'$ has at least $kr$ crossings
(and thus $kr\le c-1$) and $G_0$ has at most $c-1-kr$ crossings.

Suppose first that edges of $C_L\cup C_R$ incident with $v$ are in $G'$ drawn in the same clockwise cyclic order as in $G$.
We construct a new drawing of the graph $G$ in the following way:  Start with the drawing of $G_0$.  Take the plane drawings of $M_1$
and $M_2$ as in $G$, ``squeeze'' them and draw them very close to $M_L$ and $M_R$, respectively, so that they do not intersect any edges
of $G_0$.  Finally, draw the $r$ edges between $M_1$ and $M_2$ very close to the curve tracing $F$ (as drawn in $G'$), so that
each of them is crossed at most $k$ times.  This gives a drawing of $G$ with at most $(c-1-kr)+kr<c$ crossings, contradicting
the assumption that $G$ is $c$-crossing-critical.

\doneTodo{Referee 4: 
line 194: Please, explain the notion of "mirrored version".\\
Response: Done in a footnote. DB. }
Hence, we can assume that the edges of $C_L\cup C_R$ incident with $v$ are in $G'$ drawn in the clockwise order
$e_{M_L}$, $e_{S_L}$, $e_{M_R}$, $e_{S_R}$.  If $r=1$, then proceed analogously to the previous paragraph, except that a mirrored
version \footnote{\textsl{Mirrored version of a drawing} is the drawing obtained by reversing the vertex rotations of edges around every vertex and every crossing, and embedding the edges and the vertices accordingly. The name explains that this is homeomorphic to the original drawing seen in a mirror.} of the drawing of $M_2$ is inserted close to $M_R$; as there is only one edge between $M_1$ and $M_2$, this does
not incur any additional crossings, and we again conclude that the resulting drawing of $G$ has fewer than $c$ crossings, a contradiction.
Therefore, $r\ge 2$.

\doneTodo{Referee 4: 
line 200: ... in $v$ and the RELATIVE INTERIORS of the edges ... \\
Response: done as suggested. }
Consider the drawing $G'$, and let $q$ be a closed curve passing through $v$, following $M_L$ slightly outside $C_L$ till it meets
$F$, then following $F$ almost till it hits $M_R$, then following $M_R$ slightly outside $C_R$ till it reaches $v$.  Note that
$q$ only crosses $G_0$ in $v$ and in relative interiors of the edges, and it has at most $k$ crossings with the edges.  Shrink and mirror the part of the drawing of $G_0$
drawn in the open disk bounded by $q$, keeping $v$ at the same spot and the parts of edges crossing $q$ close to $q$; then reconnect
these parts of the edges with their parts outside of $q$, creating at most $\binom{k}{2}$ new crossings in the process.
Observe that in the resulting re-drawing of $G_0$, the path $M_L\cup M_R$ is contained in the boundary of a face (since $q$ is drawn close to it
and nothing crosses this part of $q$), and thus we can add $M$
planarly (as drawn in $G$) to the drawing without creating any further crossings.  Therefore, the resulting
drawing has at most $c-1-kr+\binom{k}{2}$ crossings.
\end{proof}

It is now easy to prove Theorem~\ref{th:boundedMD}.

\begin{proof}[Proof of Theorem~\ref{th:boundedMD}]
We prove by induction on $c$ that, for every positive integer $c\le 12$, there exists an integer $\Delta_c$ such that
every $c$-crossing-critical graph has maximum degree at most $c$.  The only $1$-crossing-critical graphs are subdivisions
of $K_5$ and $K_{3,3}$, and thus we can set $\Delta_1=4$.  Suppose now that $c\ge 2$ and the claim holds for every smaller value.
We define $\Delta_c=\max(2\Delta_{c-1},f_{\ref{cr:haveFan}}(c,6c+1))$.  Let $G$ be a $c$-crossing-critical graph and suppose
for a contradiction that $\Delta(G)>\Delta_c$.

If $G$ is not $2$-connected, then it contains induced subgraphs $G_1$ and $G_2$ such that $G_1\neq G\neq G_2$, $G=G_1\cup G_2$,
and $G_1$ intersects $G_2$ in at most one vertex.  Then $c\le \crn(G)=\crn(G_1)+\crn(G_2)$, and for every edge $e\in E(G_1)$ we have
$c>\crn(G-e)=\crn(G_1-e)+\crn(G_2)$.  Hence, $\crn(G_1)\ge c-\crn(G_2)$ and $\crn(G_1-e)<c-\crn(G_2)$ for every edge $e\in E(G_1)$, and thus $G_1$
is $(c-\crn(G_2))$-crossing-critical.  Similarly, $G_2$ is $(c-\crn(G_1))$-crossing-critical.  Since $\crn(G_1)\ge 1$
and $\crn(G_2)\ge 1$, it follows by the induction hypothesis that $\Delta(G_i)\le \Delta_{c-1}$ for $i\in\{1,2\}$, and thus
$\Delta(G)\le \Delta_c$, which is a contradiction.
\doneTodo{Referee 1: 
Page 6, line 220: $G_2$ is $(c-cr(G_1))$-crossing critical (NOT $cr(G_2)$)
\\Response: Corrected as recommended. }

Hence, $G$ is $2$-connected.  By Lemma~\ref{lemma-redraw}, there exist integers $r\ge 2$ and $k\ge 0$ such that
$kr\le c-1$ and $c-1-kr+\binom{k}{2}\ge c$, and thus $\binom{k}{2}\ge kr+1\ge 2k+1$.  This inequality is only satisfied
for $k\ge 6$, and thus the first inequality implies $c\ge kr+1\ge 13$.  This is a contradiction.  Hence, the maximum degree
of $G$ is at most $\Delta_c$.
\end{proof}

\section{Explicit $\mathbf{13}$-crossing-critical graphs with large degree}
\label{sc:construction} 

We define the following family of graphs, which is illustrated in
Figure~\ref{fig:G13-labeled}.
To simplify the terminology and the pictures, we introduce ``thick edges'':
for a positive integer $t$, we say that $uv$ is a {\em$t$-thick edge},
or an edge of {\em thickness~$t$}, if there is a bunch of $t$ parallel edges
between $u$ and $v$.
Naturally, if a $t_1$-thick edge crosses a $t_2$-thick edge, then this
counts as $t_1t_2$ ordinary crossings.
By routing every parallel bunch of edges along the ``cheapest'' edge
of the bunch, we get the following important folklore claim:

\begin{claim}\label{cl:thickedge}
For every graph $G$, 
there exists an optimal drawing $\ca D$ of $G$, 
such that every bunch of parallel edges is drawn as one thick edge in~$\ca D$.
\end{claim}

\begin{figure}
 \centering
 \begin{tikzpicture}[scale=1.3]
   \path[use as bounding box] (-5.4,-1.8) rectangle (4.4,4);
  \tikzset{every node/.style={labeled}}
  \tikzset{every path/.style={thick}}
  \node[fill=black!10!white] (x) at (-0.92,0) {$x_1$};
  \node[fill=black!10!white,
  label=below:{{\small\hspace*{-6ex}$(x_2\!=\!x_0^{k\!-\!1}\!\!=\!x_0^k)$\hspace*{-6ex}}}] (xx) at (-0.08,0) {$x_2$};
  \node (u1) at (-2,0) {$u_1$};
  \node (u2) at (-3,0) {$u_2$};
  \node (u3) at (-4,0) {$u_3$};
  \node (u4) at (-5,0) {$u_4$};
  \node[label=above:{{\small\hspace*{-6ex}($u_5=w_2^1$)}}] (u5) at (-4.5,1.5) {$u_5$};
  \node (v1) at (1,0) {$v_1$};
  \node (v2) at (2,0) {$v_2$};
  \node (v3) at (3,0) {$v_3$};
  \node (v4) at (4,0) {$v_4$};
  \node[label=above:{{\small($w_3^k=v_5$)\hspace*{-6ex}}}] (v5) at (3.5,1.5) {$v_5$};
  \node[label=left:{{\small($w_3^{k-2}\!=w^{k-1}_2$)~~~}}] (h1) at (-1,3.75) {$w^{k\!-\!1}_2$};
  \node (h2) at (-0.3,2.5) {$w^{k\!-\!1\!}_1$};
  \node (h3) at (0.75,2.44) {$w^{k\!-\!1\!}_4$};
  \node[label=right:{{\small~~($w_3^{k-1}\!=w^k_2$)}}] (k1) at (1.7,3.55) {$w^k_2$};
  \node (k2) at (1.6,2.1) {$w^{k}_1$};
  \node (k3) at (2.2,1.5) {$w^{k}_4$};
  \tikzset{every node/.style={edge-label}}
  \draw[red] (x) -- (u5) -- node {$4$} (u4) -- node {$3$} (u3) -- node {$4$}
   (u2) -- node {$5$} (u1) -- node {$7$} (x) -- node {$7$} (xx) -- node {$7$} (v1) -- node {$5$}
   (v2) -- node {$4$} (v3) -- node {$3$} (v4) -- node {$4$} (v5) -- (xx);
  \draw[blue]
   (u2) edge[bend right=70] (v3)
   (u3) edge[bend right=70] (v2)
   (u1) edge[out=300,in=205] node {2} (v4)
   (v1) edge[out=255,in=335] node {2} (u4)
  ;
  \draw (h1) edge[thick] node {2} (h2) (h2) -- (h3) (h3) edge[thick] node {2} (k1) (k1) -- (h1);
  \draw (k1) edge[thick] node {2} (k2) (k2) -- (k3) (k3) edge[thick] node {2} (v5) (v5) -- (k1);
  \draw (h2) -- (xx) -- (h3);
  \draw (k2) -- (xx) -- (k3);
  \begin{scope}[on background layer]
   \draw[fill=black!20!white] (x.center) to[bend right=16] (u5.center) -- (h1.center) to[bend right=16] cycle;
  \end{scope}
  \node[rotate=33, draw=none, fill=none] (label) at (140:3) {\huge$\cdots$};
 \end{tikzpicture}

 \caption{The graph $\ccgn{13}{k-2,2}$ of Definition~\ref{def:ccg13k}, drawn with
  $13$ crossings.
  The thick edges of this graph have their thickness written as numeric labels,
  and all the unlabeled edges are of thickness~$1$.
  The bowtie part of this graph is drawn in red and blue (where blue edges
  are those between $u_i$ and $v_j$ vertices), and the wedges are drawn in black.
  Only the $(k-1)$-th and $k$-th wedges (incident to~$x_2$) are detailed,
  while the remaining $k-2$ wedges to the left (which are in this example
  all incident to~$x_1$) analogously span the grey shaded area.}
 \label{fig:G13-labeled}
\doneTodo{Referee 4: 
Figure 1: Please, draw at least two $D_i$'s. \\
Response: done PH. }
\end{figure}

\begin{definition}[Critical family $\{\ccgn{13}{k_1,\ldots,k_m}\}$]\label{def:ccg13k}
Let $m\geq1$ and $k,k_1,k_2,\ldots,k_m\geq1$ be integers such
that~$k=k_1+\ldots+k_m$.
Let $C_u$ be a $6$-cycle on the vertex set $\{x_1,u_1,u_2,u_3,u_4,u_5\}$ 
with (thick) edges $x_1u_1$, $u_1u_2$, $u_2u_3$, $u_3u_4$, $u_4u_5$, $u_5x_1$
which are of thickness~$7,5,4,3,4,1$ in this order.
Analogously, let $C_v$ be a $6$-cycle on the vertex set
$\{x_m,v_1,v_2,v_3,v_4,v_5\}$ isomorphic to $C_u$ in this order of vertices.
Let $P_x$ be a path of length $m-1$ on the vertices $(x_1,x_2,\ldots,x_m)$
in this order and with all edges of thickness~$7$.
We denote by $B$ the graph obtained from the union $C_u\cup P_x\cup C_v$
(identifying the vertex $x_1$ of $C_u$ and $P_x$ and the vertex $x_m$ of
$C_v$ and~$P_x$) by adding edges $u_2v_3$ and $u_3v_2$,
and $2$-thick edges $u_1v_4$ and $u_4v_1$.

Let $D_i$, for $i\in\{1,\ldots,k\}$, denote the graph
on the vertex set $\{x_0^i,w_1^i,w_2^i,w_3^i,w_4^i\}$
with the edges $x_0^iw_1^i$, $x_0^iw_4^i$, $w_1^iw_4^i$, $w_2^iw_3^i$
and the $2$-thick edges $w_1^iw_2^i$ and $w_3^iw_4^i$.
From the union $B\cup D_1\cup\ldots\cup D_k$ 
we obtain the graph $\ccgn{13}{k_1,\ldots,k_m}$ via
\begin{itemize}
\item identifying $u_5$ with $w_2^1$ and $w_3^k$ with $v_5$,
\item for $i=2,3,\ldots,k$, identifying $w_3^{i-1}$ with $w_2^{i}$, and
\item for $j=1,2,\ldots,m$ and all $i$ such that
$k_1+\ldots+k_{j-1}+1\leq i\leq k_1+\ldots+k_j$,
identifying $x_0^i$ with $x_j$ of the path~$P_x$.
\end{itemize}

This definition is illustrated in Figure~\ref{fig:G13-labeled} for~$m=2$.
For reference, we will call the graph $B$ the {\em bowtie} of 
$\ccgn{13}{k_1,\ldots,k_m}$,
and the graph $D_i$ the $i$-th {\em wedge} of $\ccgn{13}{k_1,\ldots,k_m}$.
\end{definition}

\begin{observation}\label{ob:basics13n}
a) For every $m\geq1$ and $k_1,k_2,\ldots,k_m\geq1$, the graph 
$\ccgn{13}{k_1,\ldots,k_m}$ is $3$-connected and non-planar.
\\b) For $1<j<m$, the degree of the vertex $x_j$ equals~$2k_j+14$,
and the degree of the vertices $x_1$ and $x_m$ equals~$2k_1+15$
and~$2k_m+15$, respectively.
\end{observation}

In order to prove Theorem~\ref{th:construction13}, 
it is enough to consider the graph $G=\ccgn{13}{k_1,\ldots,k_m}$
for $m\geq2$ and $k_1=\dots=k_m=\lfloor d/2\rfloor$, and prove that
$\crn(G)\geq13$ and that, for every edge $e$ of $G$, we get
$\crn(G-e)\leq12$.
Before stepping into the proof, we remark that this does not hold for $m=1$
since $\crn(\ccgn{13}k)\leq12$ for all~$k$
(readers aware of the earlier conference paper \cite{bokal2019bounded} should note
that the similarly looking construction in \cite{bokal2019bounded} had the edges
$u_3u_4$ and $v_3v_4$ of thickness $4$ instead of~$3$).

\begin{lemma}\label{lem:base23}
$\crn(\ccgn{13}{1,1})=13$.
\end{lemma}

\begin{proof}
Figure~\ref{fig:G13-labeled} outlines a drawing of $\ccgn{13}{k_1,k_2}$
with $13$ crossings for all $k_1,k_2\geq1$.
For the lower bound on $\crn(\ccgn{13}{1,1})$, 
we use the computer tool {\em Crossing Number Web
Compute}~\cite{chimaniWiedera16} which uses an
ILP formulation of the crossing number problem (based on Kuratowski
subgraphs), and solves it via a branch-and-cut-and-price routine.
Moreover, this computer tool generates machine-readable proofs%
\footnote{See \url{http://crossings.uos.de/job/zS43gWV2yd-ZKmit_DNwSg}.
 Vertices~$x_1$ and $x_2$ are labeled~$0$ and $6$, respectively.
 Cycles $C_u$ and $C_v$ traverse vertices~$0,1,2,3,4,5$ and~$6,7,8,9,10,11$ in that order, respectively.
}
of the lower bound, which (roughly) consist of a branching tree 
in which every leaf holds an LP formulation of selected Kuratowski
subgraphs certifying that, in this case, the crossing number must be greater
than~$12$.
\end{proof}

\begin{remark}
Subsequently to finishing this paper, Hlin\v{e}n\'y and Korbela have found a
relatively short self-contained and computer-free proof \cite{G13handwritten}
of Lemma~\ref{lem:base23}.
\end{remark}

\begin{figure}
a)%
 \newcommand{\contour}[5]{\draw[fill=#2!50!white,draw=none,shift=(#1.center),xscale=.8333] (0,0) -- +(#3:#5) arc (#3:#4:#5);}
 \begin{tikzpicture}[xscale=1.2]
  \tikzset{every node/.style={labeled}}
  \tikzset{every path/.style={thick}}

  \node (x) at (0,0) {$x_1$};
  \node (u5) at (-3.6,2) {$u_5$};
  \node (v5) at (3.6,2) {$\!w^{k_1}_3\!$};
  \node (a1) at (-2.2,3.2){$w^i_2$};
  \node (a2) at (0,3.7){$w^i_3$};
  \node (a3) at (2.2,3.2){$\!w^{i\!+\!1\!}_3$}; 
  \node (b3) at (-1.2,2.1){$w^i_1$};
  \node (b4) at (-0.4,2.3){$w^i_4$};
  \node (b6) at (1.2,2.1){$\!w^{i\!+\!1\!}_4$};
  \node[simple] (b5) at (0.4,2.3){$\!w^{i\!+\!1\!}_1$};
  \tikzset{every node/.style={edge-label}}
 
 \begin{scope}[on background layer]
	\draw[blue!50!white,line width=5pt] (a1.center) -- (a2.center) -- (a3.center);
	\contour{a2}{blue}{0}{180}{.43}
	\draw[darkgreen!50!white,line width=5pt] (b3.center) -- (b4.center) -- (a2.center)
		 -- (b5.center)  -- (b6.center);
	\contour{b4}{darkgreen}{190}{60}{.42}
	\contour{a2}{darkgreen}{-100}{-80}{.43}
	\contour{b5}{darkgreen}{120}{-10}{.44}
 \end{scope}

\draw[red] (u5) to[bend right=10] (x);
\draw (a1) -- (a2) -- (a3);
\draw (b3) -- (x) -- (b4);
\draw (b5) -- (x) -- (b6);
\draw  (a1) edge[thick] node {2} (b3);
\draw (b3) -- (b4) edge[thick] node {2} (a2) (a2) edge[thick] node {2} (b5) (b5)--(b6);
\draw (b6) edge[thick] node {2} (a3);

\begin{scope}[on background layer]
\draw[fill=black!16!white] (x.center) to[bend right=10] (u5.center) -- (a1.center) to[bend right=10] cycle;
\end{scope}

\begin{scope}[on background layer]
\draw[fill=black!16!white] (x.center) to[bend right=10] (a3.center) -- (v5.center) to[bend right=10] cycle;
\end{scope}

 \node[rotate=35, draw=none, fill=none] (label) at (138:3) {\huge$\cdots$};
 \node[rotate=-35, draw=none, fill=none] (label) at (41:3) {\huge$\cdots$};
\end{tikzpicture}

\vspace*{-13ex}\par\hspace*{.3\hsize}
b)%
 \begin{tikzpicture}[xscale=1.2]
  \tikzset{every node/.style={labeled}}
  \tikzset{every path/.style={thick}}

  \node (x) at (0,0) {$x_1$};
  \node (u5) at (-3.6,2) {$u_5$};
  \node (v5) at (3.6,2) {$\!w^{k_1}_3\!$};
  \node (a1) at (-2.2,3.2){$w^i_2$};
  \node (a2) at (0,3.7){$w^i_3$};
  \node (a3) at (1.2,3.3){$\!w^{i\!+\!1\!}_3$}; 
  \node (b3) at (-1.2,2.1){$w^i_1$};
  \node (b4) at (-0.4,2.3){$w^i_4$};
  \node (b6) at (1.8,4.6){$\!w^{i\!+\!1\!}_4$};
  \node[simple] (b5) at (1,4.8){$\!w^{i\!+\!1\!}_1$};
  \tikzset{every node/.style={edge-label}}
 
 \begin{scope}[on background layer]
	\draw[blue!50!white,line width=5pt] (a1.center) -- (a2.center) -- (a3.center);
	\contour{a2}{blue}{-10}{190}{.43}
	\draw[darkgreen!50!white,line width=5pt] (b3.center) -- (b4.center) -- (a2.center)
		 -- (b5.center)  -- (b6.center);
	\contour{b4}{darkgreen}{190}{60}{.42}
	\contour{a2}{darkgreen}{-100}{-20}{.43}
	\contour{b5}{darkgreen}{-130}{-10}{.44}
 \end{scope}

\draw[red] (u5) to[bend right=10] (x);
\draw (a1) -- (a2) -- (a3);
\draw (b3) -- (x) -- (b4);
\draw (b5) .. controls (4,6.5) and (2,2) .. (x) edge[out=50,in=270] (b6);
\draw  (a1) edge[thick] node {2} (b3);
\draw (b3) -- (b4) edge[thick] node {2} (a2) (a2) edge[thick] node {2} (b5) (b5)--(b6);
\draw (b6) edge[thick] node {2} (a3);

\begin{scope}[on background layer]
\draw[fill=black!16!white] (x.center) to[bend right=10] (u5.center) -- (a1.center) to[bend right=10] cycle;
\end{scope}

\begin{scope}[on background layer]
\draw[fill=black!16!white] (x.center) to[bend right=10] (a3.center) -- (v5.center) to[bend right=10] cycle;
\end{scope}

 \node[rotate=35, draw=none, fill=none] (label) at (138:3) {\huge$\cdots$};
 \node[rotate=-35, draw=none, fill=none] (label) at (41:3) {\huge$\cdots$};
\end{tikzpicture}

 \caption{Two cases of vertex $w^i_3$ of the induction step in the proof of Lemma~\ref{lem:inductg}.
	In each of them we ``shrink'' two wedges into one by drawing 
	new edges $w_1^iw_4^{i+1}$ (green) and $w_2^iw_3^{i+1}$ (blue) along the depicted paths.
	In case (a), this introduces no new crossing,
	while in case (b) the new crossing between the green and the blue is
	``paid by'' a crossing which must have been on the $4$-cycle
	$(x_1,w^i_4,w^i_3,w^{i+1}_1)$ before.}
 \label{fig:induction}
\end{figure}

\begin{lemma}\label{lem:inductg}
For every $k_1\geq1$ and $k_2\geq1$, $\crn(\ccgn{13}{k_1,k_2})\geq13$. 
\end{lemma}

\begin{proof}
We proceed by induction on $k_1+k_2$, where the base case $k_1=k_2=1$ 
is proved in Lemma~\ref{lem:base23}.
Hence, we may assume that $k_1\geq2$, up to symmetry.

Consider a drawing of $\ccgn{13}{k_1,k_2}$ with $c=\crn(\ccgn{13}{k_1,k_2})$ crossings.
Let $1\leq i\leq k_1-1$.
By Claim~\ref{cl:thickedge}, we may assume that all thick edges
are drawn together in a bunch.
We now distinguish three cases based on the cyclic order of edges leaving the
vertex $w_3^i=w_2^{i+1}$ (the orientation is not important):
\begin{itemize}
\item
The edges incident to $w_3^i=w_2^{i+1}$, in a small neighbourhood of
$w_3^i$, have the cyclic order $w_3^iw_4^i$, $w_3^iw_1^{i+1}$,
$w_3^iw_3^{i+1}$, $w_3^iw_2^{i}$.
See in Figure~\ref{fig:induction}\,a), where this cyclic order is anti-clockwise.
In this case, we draw a new edge $w_1^iw_4^{i+1}$ along the path
$(w_1^i,w_4^i,w_3^i,w_1^{i+1},w_4^{i+1})$, and another new edge
$w_2^iw_3^{i+1}$ along the path $(w_2^i,w_3^i,w_3^{i+1})$
(both new edges are of thickness~$1$).
Then we delete the vertices $w_4^i,w_3^i,w_1^{i+1}$ together with incident
edges.
The resulting drawing represents a graph which is clearly isomorphic to
$\ccgn{13}{k_1-1,k_2}$ -- the wedges $i$ and $i+1$ incident to~$x_1$ 
have been replaced with one wedge.

Moreover, thanks to the assumption, we can avoid crossing between
$w_1^iw_4^{i+1}$ and $w_2^iw_3^{i+1}$ in the considered neighbourhood of
former $w_3^i$.
Therefore, every crossing of the new drawing (including possible crossings
of each of the new edges $w_1^iw_4^{i+1}$ and $w_2^iw_3^{i+1}$
among themselves or with other edges)
existed already in the original drawing of $\ccgn{13}{k_1,k_2}\!$,
and so $\crn(\ccgn{13}{k_1-1,k_2})\leq c$.
However, $\crn(\ccgn{13}{k_1-1,k_2})\geq13$ by the induction assumption, and
so $c\geq13$ holds true in this case.

\item The same proof as above works if the cyclic order around 
$w_3^i$ is $w_3^iw_4^i$, $w_3^iw_1^{i+1}$,
$w_3^iw_2^i$, $w_3^iw_3^{i+1}$.

\item 
In a small neighbourhood of $w_3^i=w_2^{i+1}$,
the incident edges have (up to orientation reversal) the cyclic order 
$w_3^iw_4^i$, $w_3^iw_3^{i+1}$, $w_3^iw_1^{i+1}$, $w_3^iw_2^{i}$.
See Figure~\ref{fig:induction}\,b).
Consider the $4$-cycle $C:=(x_1,w_4^i,w_3^i,w_3^{i+1})$ which, importantly,
uses only single edges of the $2$-thick edges incident to $w_3^i$.
In this case of the cyclic order around $w_3^i$, both sides of $C$
contain a part of the drawing of $\ccgn{13}{k_1,k_2}\!$.
Since $\ccgn{13}{k_1,k_2}-V(C)$ is connected, some edge of $C$ must be
crossed.
Consequently, the subdrawing of $\ccgn{13}{k_1,k_2}-E(C)$ has $\leq c-1$ crossings.

We finish similarly as in the first case, but within $\ccgn{13}{k_1,k_2}-E(C)$:
we draw a new edge $w_1^iw_4^{i+1}$ along the path
$(w_1^i,w_4^i,w_3^i,w_1^{i+1},w_4^{i+1})$, and another new edge
$w_2^iw_3^{i+1}$ along the path $(w_2^i,w_3^i,w_3^{i+1})$
(both new edges are of thickness~$1$, and we have so far removed only one of
the two edges of each of $w_4^iw_3^i$ and $w_3^iw_1^{i+1}$).
These two new edges mutually cross once (at most -- in case that the named
paths cross also somewhere else than at $w_3^i$, we may eliminate multiple
crossings by standard means).
After deleting the original vertices $w_4^i,w_3^i,w_1^{i+1}$,
we hence get a drawing which is again clearly isomorphic to $\ccgn{13}{k_1-1,k_2}$
and has at most $c-1+1=c$ crossings.
Since $\crn(\ccgn{13}{k_1-1,k_2})\geq13$ by the induction assumption, 
$c\geq13$ holds true also in this case.
\qedhere\end{itemize}
\end{proof}

\begin{figure}[htp]
 \centering
a)
 \begin{tikzpicture}[xscale=1.1, yscale=0.95]
   \path[use as bounding box] (-6.3,-1.9) rectangle (5.7,4.8);
  \tikzset{every node/.style={labeled}}
  \tikzset{every path/.style={thick}}
  \node[fill=black!10!white] (x) at (-1,0) {$x_1$};
  \node[fill=black!10!white] (xx) at (0,0) {$x_2$};
  \node (u1) at (-2,0) {$u_1$};
  \node (u2) at (-3,0) {$u_2$};
  \node (u3) at (-4,0) {$u_3$};
  \node (u4) at (-5,0) {$u_4$};
  \node (u5) at (-4.5,1.5) {$u_5$};
  \node (v1) at (1,0) {$v_1$};
  \node (v2) at (2,0) {$v_2$};
  \node (v3) at (3,0) {$v_3$};
  \node (v4) at (4,0) {$v_4$};
  \node (v5) at (3.5,1.5) {$v_5$};
  \node (h1) at (-1,3.75) {$w^{k\!-\!1}_2$};
  \node (h2) at (-0.3,2.5) {$w^{k\!-\!1\!}_1$};
  \node (h3) at (0.75,2.44) {$w^{k\!-\!1\!}_4$};
  \node (k1) at (1.7,3.55) {$w^k_2$};
  \node (k2) at (1.6,2.1) {$w^{k}_1$};
  \node (k3) at (2.2,1.5) {$w^{k}_4$};
  \tikzset{every node/.style={edge-label}}
  \draw[red] (x) -- (u5) -- node {$4$} (u4) -- node {$3$} (u3) -- node {$4$}
   (u2) -- node {$5$} (u1) -- node {$7$} (x) -- node {$7$} (xx) -- node {$7$} (v1) -- node {$5$}
   (v2) -- node {$4$} (v3) -- node {$3$} (v4) -- node {$4$} (v5) -- (xx);
  \draw[blue]
   (u2) edge[bend right=70] (v3)
   (u3) edge[bend right=70] (v2)
   (u1) edge[bend right=35] node {2} (v4)
  ;
    \draw[blue,thick,dashed]
     (u4) .. controls (-9.2,5.8) and (10.2,5.7) .. node {2} (v1)
     (u4) .. controls (-11.7,6.75) and (15.8,5.6) .. node {2} (v1)
    ;
  \draw (h1) edge[thick] node {2} (h2) (h2) -- (h3) (h3) edge[thick] node {2} (k1) (k1) -- (h1);
  \draw (k1) edge[thick] node {2} (k2) (k2) -- (k3) (k3) edge[thick] node {2} (v5) (v5) -- (k1);
  \draw (h2) -- (xx) -- (h3);
  \draw (k2) -- (xx) -- (k3);
  \begin{scope}[on background layer]
   \draw[fill=black!20!white] (x.center) to[bend right=16] (u5.center) -- (h1.center) to[bend right=16] cycle;
  \end{scope}
  \node[rotate=33, draw=none, fill=none] (label) at (140:3) {\huge$\cdots$};
 \end{tikzpicture}
\\[1ex]b)%
 \begin{tikzpicture}[xscale=1.1, yscale=0.95]
   \path[use as bounding box] (-5.2,-1.7) rectangle (7.1,4.2);
  \tikzset{every node/.style={labeled}}
  \tikzset{every path/.style={thick}}
  \node[fill=black!10!white] (x) at (-1,0) {$x_1$};
  \node[fill=black!10!white] (xx) at (0,0) {$x_2$};
  \node (u1) at (-2,0) {$u_1$};
  \node (u2) at (-3,0) {$u_2$};
  \node (u3) at (-4,0) {$u_3$};
  \node (u4) at (-5,0) {$u_4$};
  \node (u5) at (-4.5,1.5) {$u_5$};
  \node (v1) at (1.5,1) {$v_1$};
  \node (v2) at (3.2,2.1) {$v_2$};
  \node (v3) at (4.4,2) {$v_3$};
  \node (v4) at (5,1) {$v_4$};
  \node (v5) at (3.5,0.35) {$v_5$};
  \node (h1) at (-1,3.75) {$w^{k\!-\!1}_2$};
  \node (h2) at (-0.3,2.5) {$w^{k\!-\!1\!}_1$};
  \node (h3) at (0.75,2.44) {$w^{k\!-\!1\!}_4$};
  \node (k1) at (1.7,3.55) {$w^k_2$};
  \node (k2) at (1.6,2.1) {$w^{k}_1$};
  \node (k3) at (2.4,0.8) {$w^{k}_4$};
  \tikzset{every node/.style={edge-label}}
  \draw[red] (x) -- (u5) -- node {$4$} (u4) -- node {$3$} (u3) -- node {$4$}
   (u2) -- node {$5$} (u1) -- node {$7$} (x) -- node {$7$} (xx) -- node {$7$} (v1) -- node {$5$}
   (v2) -- node {$4$} (v3) -- node {$3$} (v4) -- node {$4$} (v5) -- (xx);
  \draw[blue]
   (u2) .. controls (6,-3.5) and (7,3) .. (v3)
   (u3) .. controls (7.5,-4.5) and (8,5) .. (v2)
   (u1) edge[out=335,in=235] node {2} (v4)
   (u4) .. controls (10,-6.5) and (9.5,10.2) .. node {2} (v1)
  ;
  \draw (h1) edge[thick] node {2} (h2) (h2) -- (h3) (h3) edge[thick] node {2} (k1) (k1) -- (h1);
  \draw (k1) edge[thick] node {2} (k2) (k2) -- (k3) (k3) edge[thick] node {2} (v5) (v5) -- (k1);
  \draw (h2) -- (xx) -- (h3);
  \draw (k2) -- (xx) -- (k3);
  \begin{scope}[on background layer]
   \draw[fill=black!20!white] (x.center) to[bend right=16] (u5.center) -- (h1.center) to[bend right=16] cycle;
  \end{scope}
  \node[rotate=33, draw=none, fill=none] (label) at (140:3) {\huge$\cdots$};
 \end{tikzpicture}
\\[1ex]c)
 \begin{tikzpicture}[xscale=1.1, yscale=0.95]
   \path[use as bounding box] (-5.3,-3.2) rectangle (6.6,3.9);
  \tikzset{every node/.style={labeled}}
  \tikzset{every path/.style={thick}}
  \node[fill=black!10!white] (x) at (-0.9,0) {$x_1$};
  \node[fill=black!10!white] (xx) at (1.5,0.35) {$x_2$};
  \node (u1) at (-2,0) {$u_1$};
  \node (u2) at (-3,0) {$u_2$};
  \node (u3) at (-4,0) {$u_3$};
  \node (u4) at (-5,0) {$u_4$};
  \node (u5) at (-4.5,1.5) {$u_5$};
  \node (v1) at (2,1.6) {$v_1$};
  \node (v2) at (3.2,2.1) {$v_2$};
  \node (v3) at (4.4,2) {$v_3$};
  \node (v4) at (5,1) {$v_4$};
  \node (v5) at (4,0.75) {$v_5$};
  \node (h1) at (-0.75,3.5) {$w^{i}_2$};
  \node (h2) at (0.3,2.25) {$w^{i}_1$};
  \node (h3) at (1,-0.75) {$w^{i}_4$};
  \node (k1) at (2.5,-2) {$w^i_3$};
  \tikzset{every node/.style={edge-label}}
  \draw[red] (x) -- (u5) -- node {$4$} (u4) -- node {$3$} (u3) -- node {$4$}
   (u2) -- node {$5$} (u1) -- node {$7$} (x) -- node {$7$} (xx) -- node {$7$} (v1) -- node {$5$}
   (v2) -- node {$4$} (v3) -- node {$3$} (v4) -- node {$4$} (v5) -- (xx);
  \draw[blue]
   (u2) .. controls (5,-7.5) and (7,3) .. (v3)
   (u3) .. controls (6.5,-8.8) and (8.2,5.5) .. (v2)
   (u1) .. controls (2,-3.8) and (4.2,-3) .. node {2} (v4)
   (u4) .. controls (9,-10.2) and (8.5,8.8) .. node {2} (v1)
  ;
  \draw (h1) edge[thick] node {2} (h2) (h2) -- (h3) (h3) edge[thick] node {2} (k1)
	 (k1) edge[out=180,in=280] (h1);
  \draw (h2) -- (xx) -- (h3);
  \begin{scope}[on background layer]
   \draw[fill=black!20!white] (x.center) to[bend right=16] (u5.center) -- (h1.center) to[bend right=16] cycle;
   \draw[fill=black!20!white] (xx.center) to[bend right=-16] (k1.center) --
(v5.center) to[bend right=-16] cycle;
  \end{scope}
  \node[rotate=29, draw=none, fill=none] (label) at (138:3) {\huge$\cdots$};
  \node[rotate=58, draw=none, fill=none] (label) at (-9:3) {\huge$\cdots$};
 \end{tikzpicture}

 \caption{Three drawings of the graph $\ccgn{13}{k_1,k_2}$ with $13$ or $14$
 crossings (where $k_2=2$ in cases (a) and (b), while $k_1=i-1$ in (c)\,).
 These drawings and their straightforward adjustments are used to argue
 criticality of the bowtie (red) edges of $\ccgn{13}{k_1,k_2}$.
 The grey areas span the crossing-free
 wedges of $\ccgn{13}{k_1,k_2}$ which are not detailed in the pictures,
 similarly as in Figure~\ref{fig:G13-labeled}.}
 \label{fig:G13-3red}
\end{figure}

\begin{figure}
 \centering\bigskip
 a)
 \begin{tikzpicture}[scale=0.68]\small
  \path[use as bounding box] (-5.4,-4.9) rectangle (4.1,4.9);
  \tikzset{every node/.style={simple}}
  \tikzset{every path/.style={thick}}
  \node[labeled,minimum size=5mm] (x) at (0,0) {$x_2$};
  \node[labeled,minimum size=5mm] (x1) at (-1,1) {$x_1$};
  \node[labeled,minimum size=5mm] (x3) at (1,-1) {$x_3$};
  \node (u1) at (0,4) {};
  \node (u2) at (-1,3) {};
  \node (u3) at (-2,2) {};
  \node (u4) at (-3,1) {};
  \node[labeled] (u5) at (-3.5,0) {$u_5$};
  \node (v1) at (0,-4) {};
  \node (v2) at (1,-3) {};
  \node (v3) at (2,-2) {};
  \node (v4) at (3,-1) {};
  \node[labeled] (v5) at (3.5,0) {$v_5$};

  \node[labeled] (wi2) at (210:3.5) {$w^i_2$};
  \node[labeled] (wi1) at ($(wi2)+(1.5,-1)$) {$w^i_1$};
  \node[labeled] (wi3) at (30:3.75) {$w^i_3$};
  \node[labeled] (wi4) at ($(wi3)+(-0.5,2)$) {$w^i_4$};

  \tikzset{every node/.style={edge-label}}

  \begin{scope}[on background layer]
   \path[draw, thick, fill=black!16!white] (u5.center) -- (wi2.center)
	 to[bend left=12] (x.center) -- (x1.center) to[bend left=18] cycle;
   \path[draw, thick, fill=black!16!white] (v5.center) -- (wi3.center)
	 to[bend left=12] (x.center) -- (x3.center) to[bend left=18] cycle;

   \draw[red] (u5) -- node {$4$} (u4) -- node {$3$} (u3) -- node {$4$} (u2) -- node {$5$} (u1)
	 -- node {$\,7$} (x1) (x3) -- node {$\,7$} (v1) -- node {$5$} (v2) -- node {$4$} (v3)
	 -- node {$3$} (v4) -- node {$4$} (v5.center) to[bend left=8] (x3.center)
	 -- node {$7$} (x.center) -- node {$7$} (x1.center) to[bend right=8] (u5);

   \draw[blue]
     (u1) .. controls ($(u1)+(-14,-7)$) and ($(v4)+(0,-8)$) .. node {2} (v4)
     (u2) .. controls ($(u2)+(-11,-5)$) and ($(v3)+(0,-6)$) .. (v3)
     (u3) .. controls ($(u3)+(-8,-3)$) and ($(v2)+(0,-4)$) .. (v2)
     (u4) to[out=190, in=180, distance=32mm] node {2} (v1)
    ;

   \coordinate (tmp) at (-.5, -4);
   \draw
    (x) -- (wi1) -- node {2} (wi2)
    (wi3) -- node {2} (wi4) -- (x)
    (wi1) .. controls ($(wi1)+(-10,0)$) and ($(wi4)+(-2.5,3.5)$) .. (wi4)
    (wi2) to[out=20, in=110, distance=2cm] (tmp) to[out=290, in=320, distance=3.5cm] (wi3)
   ;
  \end{scope}

  \node[rotate=97, draw=none, fill=none] (label) at (15:3.1) {\Large$\cdots$};
  \node[rotate=104, draw=none, fill=none] (label) at (195:3) {\Large$\cdots$};
 \end{tikzpicture}
 b)
 \begin{tikzpicture}[scale=0.68]
  \path[use as bounding box] (-5.4,-4.9) rectangle (4.2,4.9);
  \tikzset{every node/.style={simple}}
  \tikzset{every path/.style={thick}}
  \node[labeled,minimum size=5mm] (x) at (0,0) {$x_2$};
  \node[labeled,minimum size=5mm] (x1) at (-1,1) {$x_1$};
  \node[labeled,minimum size=5mm] (x3) at (1,-1) {$x_3$};
  \node (u1) at (0,4) {};
  \node (u2) at (-1,3) {};
  \node (u3) at (-2,2) {};
  \node (u4) at (-3,1) {};
  \node[labeled] (u5) at (-3.5,0) {$u_5$};
  \node (v1) at (0,-4) {};
  \node (v2) at (1,-3) {};
  \node (v3) at (2,-2) {};
  \node (v4) at (3,-1) {};
  \node[labeled] (v5) at (3.5,0) {$v_5$};

  \node[labeled] (wi2) at (210:3.5) {$w^i_2$};
  \node[labeled] (wi3) at (30:3.75) {$w^i_3$};
  \node[labeled] (wi4) at ($(wi3)+(-0.5,2)$) {$w^i_4$};
  \node[labeled] (wi1) at ($(wi4)+(-1.5,0.7)$) {$w^i_1$};

  \tikzset{every node/.style={edge-label}}

  \begin{scope}[on background layer]
   \path[draw, thick, fill=black!16!white] (u5.center) -- (wi2.center)
	 to[bend left=12] (x.center) -- (x1.center) to[bend left=18] cycle;
   \path[draw, thick, fill=black!16!white] (v5.center) -- (wi3.center)
	 to[bend left=12] (x.center) -- (x3.center) to[bend left=18] cycle;

   \draw[red] (u5) -- node {$4$} (u4) -- node {$3$} (u3) -- node {$4$} (u2) -- node {$5$} (u1)
	 -- node {$\,7$} (x1) (x3) -- node {$\,7$} (v1) -- node {$5$} (v2) -- node {$4$} (v3)
	 -- node {$3$} (v4) -- node {$4$} (v5.center) to[bend left=8] (x3.center)
	 -- node {$7$} (x.center) -- node {$7$} (x1.center) to[bend right=8] (u5);

   \draw[blue]
     (u1) .. controls ($(u1)+(-14,-7)$) and ($(v4)+(0,-8)$) .. node {2} (v4)
     (u2) .. controls ($(u2)+(-11,-5)$) and ($(v3)+(0,-6)$) .. (v3)
     (u3) .. controls ($(u3)+(-8,-3)$) and ($(v2)+(0,-4)$) .. (v2)
     (u4) to[out=190, in=180, distance=32mm] node {2} (v1)
    ;

   \coordinate (tmp) at (0, -5);
   \draw
    (x) -- (wi1) to[out=165,in=150, distance=4cm] node[near start] {2} (wi2)
    (wi3) -- node {2} (wi4) -- (x)
    (wi1) -- (wi4)
    (wi2) to[out=-25, in=180] (tmp) to[out=0, in=320] (wi3)
   ;
  \end{scope}

  \node[rotate=100, draw=none, fill=none] (label) at (15:3.1) {\Large$\cdots$};
  \node[rotate=105, draw=none, fill=none] (label) at (195:3) {\Large$\cdots$};
 \end{tikzpicture}

 \caption{Two drawings of the graph $\ccgn{13}{k_1,k_2,k_3}$, having (a) $13$ and (b)
 $18$ crossings. These drawings are used to argue criticality of edges
 of the $i$-th wedge. The grey areas span the crossing-free wedges of
 $\ccgn{13}{k_1,k_2,k_3}$ which are not detailed in the pictures,
 similarly as in Figure~\ref{fig:G13-3red}.}
 \label{fig:G13-wedge-split}
\end{figure}

\begin{theorem}\label{thm:mainalt}
For every integers $m\geq2$ and $k_1,k_2,\ldots,k_m\geq1$, the graph
$\ccgn{13}{k_1,\ldots,k_m}$ is $13$-crossing-critical.
\end{theorem}

\begin{proof}
Let $G=\ccgn{13}{k_1,\ldots,k_m}$ and 
$P_x$ be the $7$-thick path on the vertices $(x_1,x_2,\ldots,x_m)$
from Definition~\ref{def:ccg13k}.
We first prove that $\crn(G)\geq13$.
Using Claim~\ref{cl:thickedge}, at most one edge of $P_x$ is crossed, 
or we already have $14$ crossings.
So assume that all edges of $P_x$ except possibly $x_jx_{j+1}$ have no crossing.
Contracting the edges $E(P_x)\setminus\{x_jx_{j+1}\}$ thus creates no new
crossing and results in a valid drawing isomorphic to
$\ccgn{13}{l_1,l_2}$ where $l_1=k_1+\ldots+k_j$ and
$l_2=k_{j+1}+\ldots+k_m$.
We conclude with $\crn(G)\geq\crn(\ccgn{13}{l_1,l_2})\geq13$ by
Lemma~\ref{lem:inductg}.

Regarding criticality, our proof strategy can be described as follows.
We provide a collection of drawings of our graph $G$, such that each edge
$e$ of $G$ in some of the drawings, when deleted, exhibits a ``drop'' 
of the crossing number below~$13$; that is $\crn(G-e)\leq12$.
Note that, for thick edges, we are deleting only one edge of the multiple bunch.

We start with the edges of the bowtie of $G$.
For the blue edges (i.e., $u_2v_3,u_3v_2,u_1v_4,u_4v_1$), this follows
immediately from the drawing in Figure~\ref{fig:G13-labeled}
in which deleting any blue edge saves crossings.
Furthermore, one can easily split the vertices $x_1$ and $x_2$ 
in the picture to produce the full path $P_x$ as needed.
For the remaining, red bowtie edges, criticality is witnessed by the
three drawings in Figure~\ref{fig:G13-3red}.
In the first one (a), which is almost the same as Figure~\ref{fig:G13-labeled},
two alternate routings of the edge $u_4v_1$ show
criticality of the edges $x_2v_5$ and $v_4v_5$, respectively.
We symmetricaly argue about the edges $x_1u_5$ and $u_4u_5$.
The second one (b) shows criticality of the edge $v_1v_2$.
However, by pulling $v_1$ in this picture away from $x_2$ we also certify
criticality of $x_2v_1$, and by pulling $v_2$ or also $v_3$ towards $x_2$ we
get criticality of $v_2v_3$ and $v_3v_4$.
Again, we can easily split the vertices $x_1$ and $x_2$ in the drawings to
produce the path $P_x$ as needed and without further crossings.
The edges $x_1u_1$, $u_1u_2$, $u_2u_3$ and $u_3u_4$ are symmetric, too.

Consider now a red edge $x_jx_{j+1}$ of $P_x$.
Let the first wedge incident to $x_{j+1}$ be the $i$-th wedge~$D_i$.
We twist the picture from Figure~\ref{fig:G13-labeled} at the
edge $x_jx_{j+1}$, such that the wedges preceding $D_i$ stay above the path
$P_x$, and the wedges succeeding $D_i$ are now below $P_x$.
This is illustrated for $j=1$ in Figure~\ref{fig:G13-3red}(c).
The wedge $D_i$ now crosses the $7$-thick edge $x_jx_{j+1}$, 
giving a drawing of $G$ with $14$ crossings, and so certifying criticality
of the edge $x_jx_{j+1}$, since deleting it drops the number of crossings in
this drawing down to~$12$.

We are left with the last, and perhaps most interesting, cases in which $e$
is an edge in the $i$-th wedge $D_i$.
We consider a twist of the drawing of $G$ similar to that in
Figure~\ref{fig:G13-3red}(c), but this time with the wedge $D_i$ crossing
the blue bowtie edges (and itself).
This gives a drawing with $13$ crossings involving the edges
$x_2w_1^i$, $w_1^iw_4^i$ and $w_2^iw_3^i$,
which is illustrated in Figure~\ref{fig:G13-wedge-split}(a).
Hence the listed edges, and the edge $x_2w_4^i$ by symmetry,
are also critical in $G$, as desired.
Finally, we deal, up to symmetry, with the $2$-thick edge $w_1^iw_2^i$.
A slight adjustement of the last drawing gives a drawing illustrated in
Figure~\ref{fig:G13-wedge-split}(b) with exactly $18$ crossings
which are between the blue edges and $w_2^iw_3^i$, $w_1^iw_2^i$.
Since deleting one edge from the $2$-thick edge $w_1^iw_2^i$ drops the
number of crossings again down to $12$, we have shown also criticality of
$w_1^iw_2^i$ and the proof is finished.
\end{proof}

Theorem~\ref{th:construction13} is now established 
for $k_1=\ldots=k_m=\lfloor d/2\rfloor$.

\section{Extended crossing-critical construction}
\label{sc:extended} 

In the previous section, we have constructed an infinite family of
$13$-crossing-critical graphs with unbounded maximum degree.
The construction leaves a natural question about analogous
$c$-crossing-critical families for $c>13$.


Clearly, the disjoint union of the graph from Theorem~\ref{thm:mainalt}
with $c{-}13$ disjoint copies of $K_{3,3}$ yields a (disconnected) 
$c$-crossing-critical graph
with maximum degree greater than $d$, for every $c\ge 14$.
Though, our aim is to preserve also the $3$-connectivity property
of the resulting graphs.

First, to motivate the coming construction, we recall that the zip product of
Definition~\ref{def:zip} requires a vertex of degree $3$ in the considered graphs.
However, the graphs of Definition~\ref{def:ccg13k} have no such
vertex, and so we come with the following modification.

\begin{figure}[ht]
 \centering\medskip
  \begin{tikzpicture}
  \tikzset{every node/.style={labeled}}
  \tikzset{every path/.style={thick}}
  \node (v2) at (1,0) {$t_1$};
  \node (v3) at (3.2,0) {$s$};
  \node (v4) at (5,0) {$t_2$};
  \node (u2) at (-0.3,-1) {$t_3$};
  \tikzset{every node/.style={edge-label}}
  \draw[red] (v2) -- node {$h\!+\!1$} (v3) -- node {$h$} (v4) ;
  \draw[red] (v2) to (0.4,0.2);
  \draw[blue] (u2) edge[bend right=40] (v3) ;
  \draw[blue] (v2) to (0.4,-0.2);
  \end{tikzpicture}
\raise7ex\hbox{\Large$\quad\leadsto\quad$}
  \begin{tikzpicture}
  \tikzset{every node/.style={labeled}}
  \tikzset{every path/.style={thick}}
  \node (v2) at (1,0) {$t_1$};
  \node (v31) at (3.5,0) {$s$};
  \node (v33) at (2.3,-1) {$s'$};
  \node (v4) at (5,0) {$t_2$};
  \node (u2) at (-0.3,-1) {$t_3$};
  \tikzset{every node/.style={edge-label}}
  \draw[red] (v2) -- node {$h$} (v31) -- node {$h$} (v4) ;
  \draw[red] (v2) to (0.4,0.2);
  \draw[red] (v31) -- (v33) -- (v2) ;
  \draw[blue] (u2) edge[bend right=30] (v33) ;
  \draw[blue] (v2) to (0.4,-0.2);
  \end{tikzpicture}

 \caption{An illustration of the operation of locally introducing
a vertex ($s'$) of degree $3$ from Lemma~\ref{lem:plus3deg}.
This operation can be applied, e.g., to vertices
$t_1=v_1$, $s=v_2$, $t_2=v_3$, and $t_3=u_3$ of Figure~\ref{fig:G13-labeled}.}
 \label{fig:v3-split}
\end{figure}

\begin{lemma}\label{lem:plus3deg}
Assume a graph $H$ with vertices $t_1,t_2,t_3$ and $s$ such that
\begin{itemize}
\item[a)]
vertex $s$ has no more neighbours than $t_1,t_2,t_3$ in $H$,
the edge $t_1s$ is $(h+1)$-thick, $t_2s$ is $h$-thick, $t_3s$ is $1$-thick,
and
\item[b)]
vertex $t_1$ is of degree at most $h+5$ in $H$, or there is a neighbour
$w\not=s$ of $t_1$ such that $t_1$ is of degree at most $h+3$ in $H-t_1w$.
\end{itemize}
Other edges of $H$ are not important. 

Let $H'$ be created by making the edge $t_1s$ only $h$-thick,
deleting the edge $t_3s$,
and adding a new vertex $s'$ adjacent via three $1$-thick edges to
the vertices $t_1,t_3$ and~$s$.
See Figure~\ref{fig:v3-split}. Then $\crn(H')\geq\crn(H)$.
Furthermore, if $H$ is a $c$-crossing-critical graph and
$\crn(H'-ss')<c$, then $H'$ is also $c$-crossing-critical.
\end{lemma}

\begin{proof}
Assume a drawing $\ca D$ of $H'$ with $\crn(H')$ crossings.
By Claim~\ref{cl:thickedge}, we have $st_1$ and $st_2$ drawn each as one
thick edge.
We consider two cases based on the crossings on $ss'$ in $\ca D$.

First, there are at least $2$ crossings on $ss'$ in $\ca D$.
We modify $\ca D$ to $\ca D'$ as follows:
delete the current edge $ss'$, and pull the vertices $s$ and $s'$ along their
edges to $t_1$ so that no crossing remains on $st_1$ and $s't_1$ in $\ca D'$.
This modification does not change the number of crossings on the paths
$(t_2,s,t_1)$ and $(t_3,s',t_1)$.
Then draw a new ($1$-thick) edge $ss'$ in $\ca D'$ closely along the path
$(s',t_1,s)$, crossing only some of the edges incident with~$t_1$
(and choosing ``the better side'' of $t_1$).
Thanks to the assumption (b), this makes only at most $2$ crossings on $ss'$
in $\ca D'$:  if $t_1$ is of degree $h+5$ then we cross at most
$[(h+5)-(h+1)]/2=2$, and if $t_1$ is of degree $h+3$ in $H-t_1w$,
then we can avoid crossing $t_1w$ and again cross at most $(h+3)-(h+1)=2$.

Altogether, there are no more crossings in $\ca D'$ than there were in~$\ca D$.
Since $st_1$ is crossing-free, we can turn $st_1$ into an $(h+1)$-thick edge
and still have at most $\crn(H')$ crossings.
Then we delete the edge $s't_1$ and obtain a subdivision of the graph $H$
with at most $\crn(H')$ crossings,
which certifies $\crn(H')\geq\crn(H)$.

Second, we assume that there is at most $1$ crossing on $ss'$ in $\ca D$.
Let the number of crossings on each edge of the
parallel bunch $st_1$ be $a$ and on the edge $s't_1$ let it be $b$.
If $b\geq a$, then we do the same as previously:
delete the edge $s't_1$ and turn $st_1$ into an $(h+1)$-thick edge.
The resulting drawing is a subdivision of $H$
and the new number of crossings is $\crn(H')-b+a\leq\crn(H')$,
again certifying $\crn(H')\geq\crn(H)$.

Otherwise, if $b\leq a-1$, there are altogether at most $b+1\leq a$ crossings
along the ($1$-thick) path $(t_1,s',s)$.
We hence make no more crossings than $\crn(H')$ if we redraw the $h$-thick
edge $t_1s$ closely along the path $(t_1,s',s)$ and ``through'' the
vertex~$s'$, creating a subdivision of a graph isomorphic to $H$ 
(now with $s$ subdividing $h$-thick edge~$s't_2$).
Again, the conclusion is that $\crn(H')\geq\crn(H)$.

The last bit is to argue $c$-crossing-criticality of $H'$ 
under the additional assumption of the lemma.
Consider any edge $e\in E(H')\cap E(H)$, and a drawing $\ca D$
of $H-e$ with fewer than $c$ crossings.
Since $s$ has only three neighbours in $H$, the vertex $s'$ can be chosen in
$\ca D$ as subdividing a suitable one of the edges of the $(h+1)$-thick
bunch $t_1s$, the one consecutive to $t_3s$ in the rotation around $s$ 
in~$\ca D$.
This results in a drawing of $H'-e$ with same number of crossings (fewer
than~$c$).
It remains to consider the edges of $E(H')\setminus E(H)=\{t_1s',t_3s',ss'\}$.
We have got the assumption $\crn(H'-ss')<c$, and drawings of
$H'-t_1s'$ and of $H'-t_3s'$ with fewer than $c$ crossings are subdivisions
of the corresponding drawings of $H-t_1s$ and~$H-t_3s$.
\end{proof}

\begin{proof}[Proof of Corollary~\ref{cr:construction}]
Similarly as in the previous section, we take the $13$-crossing-critical graph
$G=\ccgn{13}{k_1,\ldots,k_m}$ of Theorem~\ref{thm:mainalt}
for $k_1=\ldots=k_m=\lfloor d/2\rfloor$.
Then we apply Lemma~\ref{lem:plus3deg} to the vertices
$t_1=v_1$, $s=v_2$, $t_2=v_3$ and $t_3=u_3$ of~$G$.
This results in a graph $G'$ having a vertex $s'$ of degree~$3$.
Moreover, since $\crn(G'-ss')\leq11$ which can be easily seen from
Figure~\ref{fig:G13-labeled} (we avoid crossings with $u_4v_1$),
we get that $G'$ is $13$-crossing-critical.

Hence let $G(13,d,m)=G'$.
For $c>13$, we proceed by induction, assuming that we have already
constructed the graph $G(c-1,d,m)$ and it contains a vertex of degree~$3$.
Theorem~\ref{thm:zip} establishes that $G(c,d,m)$, as a zip product of 
$G(c-1,d,m)$ with $1$-crossing-critical $K_{3,3}$, is $c$-crossing-critical.
Furthermore, $G(c,d,m)$ again contains a vertex of degree~$3$ coming from
the $K_{3,3}$ part.
\end{proof}

\section{Concluding remarks and open problems}
\label{sc:conclusion}

While our contribution closes the questions related to the 
validity of the bounded maximum degree conjecture, the following
natural problems remain open:

\begin{problem}
For each $c\le 12$, determine the least integer $D(c)$ bounding maximum degree of 
$c$-crossing-critical graphs.
\end{problem}

\begin{problem}
Develop a theory of wedges that parallels the theory of tiles
(cf.~\cite{pinontoanRichter03})
for constructively establishing $c$-crossing-criticality of graphs with 
large maximum degrees.
\end{problem}

Note that with our construction we can get arbitrarily repeated even degrees in
$G=\ccgn{13}{k_1,\ldots,k_m}$, cf.~Observation~\ref{ob:basics13n}\,b),
but only two large-odd-degree vertices there.
In general, vertices of high odd degrees in $c$-crossing critical graphs seem to rely 
on some local property of the graph, unlike even degrees that can rely
simply on sufficiently many relevant edge-disjoint paths passing through the vertices. 
Indeed, the only other known examples of large odd degrees in infinite families of $c$-crossing-critical graphs are related to staircase-strip tiles \cite{bokalBracicDernarHlineny19}. Hence we suggest also the following question:
\begin{problem}
Does there exist, for some/any $c\geq13$, 
a family of $c$-crossing-critical graphs, such that for a prescribed set $O$ of odd 
integers greater than $3$ and each integer $m$, the family would contain a graph with at least $m$ vertices of each degree in $O$?
\end{problem}

Furthermore, Lemma~\ref{lem:plus3deg} can be applied iteratively to selected vertices of 
each wedge of the graphs $G=\ccgn{13}{k_1,\ldots,k_m}$ to 
produce new $c$-crossing-critical graphs which would be $3$-connected and 
have no double edges within the wedges. 
However, removing the remaining multiple edges in the bowtie subgraph 
would require a different approach. Hence, our final problem is:
\begin{problem}
For which $c$ does there exist a family of $3$-connected simple $c$-crossing-critical graphs containing vertices of arbitrarily large degree?
\end{problem}

%

\doneTodo{Referee 3:
L412, ref.1 :  Be consistent about whether or not you abbreviate or give
  full first names of authors. \\
Response: TODO. Tilo, IIRC, you did the bibliography. Could you, please, fix this?}
\begin{small}
\bibliography{main}
\end{small}

\newpage\appendix

\end{document}